\documentclass[10pt]{article}
\pdfoutput=1

\usepackage[margin=1in]{geometry}
\usepackage[english]{babel}
\usepackage[utf8]{inputenc}
\usepackage[T1]{fontenc}
\usepackage{lmodern}
\usepackage{microtype}
\usepackage[font=small]{caption}
\usepackage[font=small]{subcaption}
\usepackage{amsmath,amssymb,amsthm}
\usepackage[inline]{enumitem}
\usepackage{graphicx}
\usepackage{multirow}
\usepackage[hidelinks]{hyperref}
\usepackage{xspace}
\usepackage[numbers,sort,longnamesfirst]{natbib}
\usepackage{doi,url}

\usepackage{xcolor}
\usepackage{authblk}
\usepackage{tikz}
\usetikzlibrary{calc,arrows.meta,shapes.geometric,decorations.pathreplacing}

\definecolor{figblue}{RGB}{0,90,181}
\definecolor{figred}{RGB}{220,50,32}
\definecolor{figgray}{RGB}{130,130,130}
\tikzset{
	qnode/.style={regular polygon, regular polygon sides=4, draw=black,
		fill=black, inner sep=2pt},
	cand/.style={circle, draw=figgray, fill=white, inner sep=2.1pt},
	selnode/.style={circle, draw=figblue!80!black, fill=figblue, inner sep=2.4pt},
	goldnode/.style={circle, draw=figblue!80!black, fill=figblue, inner sep=2.4pt},
	dupnode/.style={circle, draw=figred, densely dotted, fill=figred!15,
		inner sep=2.1pt},
	activenode/.style={circle, draw=figred!80!black, fill=figred, inner sep=2.2pt},
	paneltitle/.style={font=\small\itshape},
	tinylabel/.style={font=\scriptsize},
}

\theoremstyle{plain}
\newtheorem{prop}{Proposition}[section]
\newtheorem{cor}[prop]{Corollary}
\newtheorem{thm}[prop]{Theorem}
\theoremstyle{definition}
\newtheorem{defn}[prop]{Definition}
\theoremstyle{remark}

\numberwithin{equation}{section}

\newcommand{\R}{\mathbb{R}}
\newcommand{\E}{\mathbb{E}}
\newcommand{\RNG}{\textsc{RNG}\xspace}
\newcommand{\MMR}{\textsc{MMR}\xspace}
\newcommand{\DPP}{\textsc{DPP}\xspace}
\newcommand{\ANN}{\textsc{ANN}\xspace}
\newcommand{\RNGSCORE}{\textsc{RNG-Score}\xspace}
\newcommand{\score}[1]{\operatorname{score}_{#1}}
\newcommand{\similarity}[1]{\operatorname{sim}_{#1}}
\newcommand{\CE}{\text{CE}}
\newcommand{\hinge}[1]{\left[#1\right]_+}

\hypersetup{
	pdftitle={Finding the Right Balance: Relevance and Diversity in LLM Retrieval},
	pdfauthor={Guillaume Brouillette, Faustin Kagabo, Usef Faghihi, Nadia Ghazzali},
}

\title{Finding the Right Balance: Relevance and Diversity in LLM Retrieval}

\author[1]{Guillaume Brouillette\thanks{\texttt{Guillaume.Brouillette@uqtr.ca}}}
\author[1]{Faustin Kagabo\thanks{\texttt{Faustin.Kagabo@uqtr.ca}}}
\author[1]{Usef Faghihi\thanks{\texttt{Usef.Faghihi@uqtr.ca}}}
\author[1]{Nadia Ghazzali\thanks{\texttt{Nadia.Ghazzali@uqtr.ca}}}
\affil[1]{D\'epartement de math\'ematiques et d'informatique,\protect\\ Laboratoire d'intelligence artificielle appliqu\'ee (LI2A),\protect\\ INSTINT (Informatique, Statistiques, et Technologies Intelligentes Transdisciplinaires),\protect\\ Universit\'e du Qu\'ebec \`a Trois-Rivi\`eres, Trois-Rivi\`eres, QC, Canada}

\date{October 6, 2026}

\begin{document}

	\maketitle

	\begin{abstract}
		Retrieval diversification is widely available in retrieval-augmented generation (RAG) frameworks, yet prior studies disagree on whether it improves retrieval and answer quality. We show that its effectiveness varies primarily with candidate-pool redundancy, in a pattern consistent with the number of distinct evidence pieces a query requires. Using controlled near-duplicate injection and production-style overlapping chunking, we find that diversification harms relevance, evidence coverage and answer quality on clean pools, but becomes beneficial on multi-evidence tasks when redundancy causes nearest-neighbor retrieval to select repeated passages. We therefore introduce a query-adaptive rule that diversifies only when the effective number of distinct documents in the nearest-neighbor top-$k$ selection falls below the query's evidence requirement. Computed from existing embeddings, the rule captures most of the achievable gain, transfers across datasets and encoders and automatically reduces to nearest-neighbor retrieval for single-evidence queries. We also introduce \RNGSCORE, a geometric reranker with an exact nearest-neighbor fallback whose margin indicates duplicate structure. Overall, we conclude that diversification should be used selectively, based on observable redundancy and evidence requirements.
	\end{abstract}

	\section{Introduction}

	Retrieval-Augmented Generation (RAG) augments a language model with external documents retrieved at inference time \cite{lewis2020rag,guu2020realm,izacard2021fid,borgeaud2022retro,shi2023replug,asai2023selfrag}. In its simplest form, retrieval returns the top-$k$ nearest documents to a query embedding. A long-standing concern with this strategy is redundancy: several retrieved passages may say nearly the same thing and thereby waste a context budget that language models already use unevenly \cite{liu2024lostmiddle}. This concern has motivated a growing line of work that injects an explicit diversity objective into the retrieval or reranking stage of RAG pipelines \cite{rezaei2025vendirag,khan2026dfrag,corteslemos2025drag,li2024smartrag,lu2026ccbqp,adagres2025,dppscaling2026}. In the same spirit, diversification heuristics such as Maximal Marginal Relevance (\MMR) \cite{carbonell1998mmr} ship as built-in one-parameter retrieval modes in widely used frameworks,\footnote{For example, \texttt{search\_type="mmr"} in LangChain vector-store retrievers and \texttt{vector\_store\_query\_mode="mmr"} in LlamaIndex. In both frameworks plain similarity search is the default and \MMR\ is enabled by a single switch, with no guidance on when it should be.} so that a practitioner can toggle diversification on.

	However, it is unclear whether diversifying document retrieval in RAG pipelines helps or harms. In fact, many works on the subject reach different conclusions. Diversity-aware pipelines report end-to-end gains of several points on multi-hop question answering \cite{rezaei2025vendirag,khan2026dfrag}, while controlled evaluations of retrieval configurations find that \MMR-style diversification yields no significant improvement over plain nearest-neighbor retrieval \cite{eibich2024aragog,bal2026biomedrag}.

	These diverging conclusions are reached in different contexts where the retrieval diversity is given \emph{instrumental} credit, meaning that diversity is not the end goal but rather a means to an end, as in a RAG pipeline. When an evaluation grants diversity \emph{intrinsic} credit, namely when its metric suite rewards diversity directly, as the beyond-accuracy metrics of recommendation and search-result diversification do \cite{clarke2008alphaNDCG,vargas2011rank,kaminskas2017diversity}, then diversifying retrieval improves results by design.

	\begin{figure}[t]
		\centering
		\includegraphics[width=0.72\linewidth, alt={Line plot of retrieval coverage (S-Recall at 5) against the measured near-duplicate pair fraction on a logarithmic axis for three policies on HotpotQA fullwiki: always k-NN, always diversify with MMR at lambda 0.7, and the decision rule with tau equal to h. Always k-NN is highest on clean pools and collapses beyond a shaded crossover bracket, always diversifying stays nearly flat at a lower level, and the decision rule follows the upper of the two curves across the whole range.}]{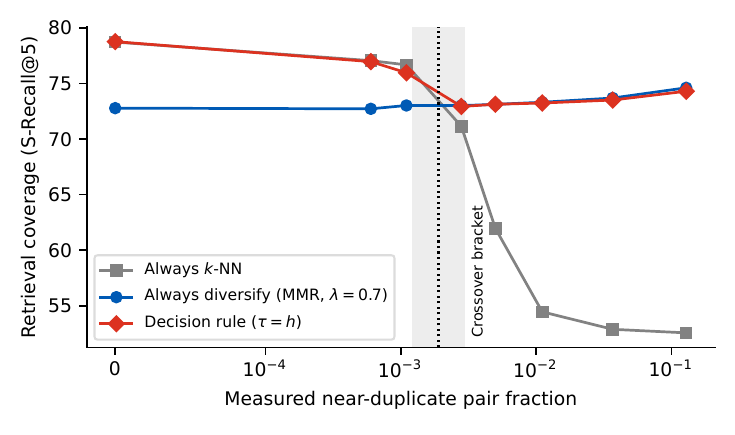}
		\caption{The paper's central result: coverage of three policies on HotpotQA fullwiki as measured pool redundancy grows under controlled injection. Nearest-neighbor selection, the framework default, collapses past the shaded crossover bracket (measured redundancy between $0.0012$ and $0.0030$, Section~\ref{sec:rq2}). Diversifying hurts on clean pools and the decision rule of Section~\ref{sec:rule} tracks the upper envelope of both.}
		\label{fig:regimes-summary}
	\end{figure}

	This paper analyzes the impact of diversity at the retrieval stage in a RAG pipeline. Thus, it is scoped to a context where diversity is granted instrumental credit. Within that scope we organize the question mainly around one observable variable, namely the redundancy of the candidate pool. This leads to a diagnostic computed from the system's own embeddings which, thresholded at the query's evidence requirement, suffices to decide whether to diversify. We also quantify how much of the available gain that rule captures. Figure~\ref{fig:regimes-summary} summarizes the result.

	Furthermore, the study needs a diversifier predictable enough to serve as a stable probe across regimes. We use a simple geometric construction inspired by the concept of relative neighborhood graph (\RNG) \cite{toussaint1980rng}, namely the \RNGSCORE (see Figure~\ref{fig:knn-vs-rng}). At both extremes of its margin parameter $\gamma$, this score reduces exactly to nearest-neighbor ranking, so validated tuning over a grid containing such a margin can always express ``do not diversify''. This exact two-sided fallback anchors the \RNGSCORE to the relevance-only baseline, and its validation-tuned margin additionally provides an offline readout of the pool's duplicate structure.

	The contributions of this paper include:
	\begin{enumerate}[label=$\mathrm{(\arabic*)}$]
		\item\label{enum:decisionrule} a \emph{decision rule} which diversifies when the effective number of distinct documents in the nearest-neighbor top-$k$ falls below the number of evidence pieces the query requires, keeps its clean-pool cost within a one-point tolerance on most evaluated sweeps, captures most of the achievable adaptive gain, transfers with no retuning across encoders, datasets, redundancy mechanisms and selection sizes and switches itself off on single-evidence retrieval where diversification has nothing to win;
		\item a regime analysis of retrieval diversification organized around pool redundancy, with protocols that measure redundancy on the embedded pool and manipulate it directly (near-duplicate injection, production-style overlapping chunking, a weak-versus-strong pipeline contrast), all reported on one measured-redundancy axis;
		\item a controlled, seed-replicated study of diversified reranking on fixed pools across multi-hop QA, open-domain QA and BEIR, showing that on clean pools the classical diversifiers lower relevance and coverage-aware diversity in nearly every case, significantly so on the primary coverage endpoint of the headline benchmark;
		\item per-query oracle bounds certifying that on clean pools, no policy routing among the evaluated methods and operating points can gain more than marginally, with headroom opening in proportion to measured redundancy;
		\item misspecification downside, the worst shortfall of a method against the $k$-NN default over its hyperparameter range when the regime is unknown, as an evaluation criterion measured for every method; the \RNGSCORE provably reduces to $k$-NN at both margin extremes, the negative end of its default grid is verified to reproduce $k$-NN on the evaluated clean pools, and its tuned margin doubles as an offline diagnostic of duplicate structure.
	\end{enumerate}

	The decision rule \ref{enum:decisionrule} is the most practical contribution: the crossover from the clean to the redundant regime is detectable per query, from the retrieval system's own embeddings and the task's evidence requirement, without relevance labels and without knowing which regime one is in.

	\section{Background and Related Work}
	\label{sec:background}

	\subsection{Modern RAG pipelines}

	Modern RAG systems decompose retrieval into stages. A first-stage retriever produces a candidate pool by sparse, dense or hybrid retrieval \cite{robertson2009bm25,karpukhin2020dpr,xiong2021ance,thakur2021beir}, typically accelerated with \ANN infrastructure \cite{johnson2019faiss,malkov2018hnsw}. The pool is usually refined before generation using stronger rerankers, cross-encoders or late-interaction models \cite{nogueira2019bertrerank,khattab2022colbertv2,sharma2025ragfrontiers,akarsu2026bm25corrective}. Encoders range from compact sentence-embedding models \cite{reimers2019sbert} to stronger families such as BGE-M3 and Qwen3-Embedding \cite{Chen2024,Zhang2025Qwen3,xiao2023bge} and proprietary services. Agentic, corrective, graph-based and chain-of-retrieval paradigms iterate the retrieval step \cite{singh2025agenticrag,yan2024crag,edge2024graphrag,li2025corag}. Diversification is orthogonal to all of these choices: every such pipeline ultimately selects $k$ documents from a finite pool and the question is whether that selection should be diversified.

	\subsection{The case for diversity in RAG}
	\label{sec:prodiv}

	Classical information retrieval offers several diversification mechanisms: \MMR greedily trades relevance against similarity to already selected items \cite{carbonell1998mmr}; \DPP-based methods model repulsion probabilistically \cite{kulesza2012dpp,chen2018fastdpp,wilhelm2018practical}; Maxmin and related dispersion objectives maximize spread directly \cite{ravi1994maxmin,borodin2017beyond,gollapudi2009axiomatic}; and intent-aware approaches diversify over explicit query aspects \cite{agrawal2009diversifying,santos2010exploiting}. Moreover, work has been done to develop selective diversification in the context of web search results \cite{santos2010selective} based on specific characteristics of a given query, such as ambiguity \cite{cronen2002predicting}.

	A recent line of work argues that RAG specifically benefits from diversity-aware retrieval. Vendi-RAG \cite{rezaei2025vendirag} couples a Vendi-score retrieval objective \cite{friedman2023vendi} with an iterative pipeline in which a LLM judge scores candidate answers and the diversity weight is adapted across rounds. It reports accuracy gains of up to $4.2\%$ on HotpotQA over Adaptive-RAG baselines \cite{jeong2024adaptiverag}. DF-RAG \cite{khan2026dfrag} builds a geometric \MMR variant and selects a per-query diversity level with planner and evaluator LLMs. It reports F1 gains of $4$--$10\%$ over vanilla RAG with an estimated per-query oracle ceiling around $18\%$. Retriever portfolios extend the adaptive idea to the retriever level, training a lightweight router that selects per query among a small complementary set of retrievers including diversity-oriented members \cite{stouras2026portfolios}. Submodular set functions have been applied to the reranking stage to balance relevance and coverage \cite{corteslemos2025drag}, alongside determinantal-point-process selection \cite{li2024smartrag} and relevant-information-gain criteria \cite{pickett2024infogain}, and scalable set-selection formulations have been proposed for the underlying combinatorial problem \cite{lu2026ccbqp,dppscaling2026}. Adjacent adaptive mechanisms let the pipeline decide how many documents to keep per query \cite{meng2025dps,xu2025car} and ambiguous-question pipelines retrieve across inferred interpretations of the query \cite{in2025diva}. A distinct, query-side notion of diversity decomposes queries by lexical diversity for finer-grained relevance assessment \cite{zhang2025drag} or replaces inter-document dissimilarity altogether with coverage of a sub-query demand distribution \cite{zhang2026georag}. Diversity has also been found beneficial for open-ended generation where comprehensiveness or answer diversity is itself the target \cite{salemi2025planrefine,diverge2026} and for long-context selection \cite{wang2025diversityllm}. Notably, the most recent entries in this line motivate diversification explicitly by pool redundancy: AdaGReS \cite{adagres2025} targets redundant and near-duplicate chunks that waste a token budget. Pool redundancy is also emerging as an evaluation variable in its own right: standard benchmarks measure far cleaner than the finance, legal and patent corpora on which strong retrievers degrade \cite{cho2026rare} and exact deduplication pays off only outside the clean academic regime \cite{schelpe2026dedup}, though neither line evaluates the selection step itself.

	The experimental regimes behind these positive results deserve attention: Vendi-RAG entangles diversification with extra LLM calls and never isolates the selection step on a fixed pool, while DF-RAG retrieves from long contexts chunked into fixed-length windows, a pool redundant by construction, and its headline gains come from per-query adaptive selection of the diversity level rather than from fixed diversification. The set-selection and reranking methods \cite{corteslemos2025drag,li2024smartrag,pickett2024infogain,adagres2025} do isolate the diversification step on a fixed candidate pool, but not in the regime studied here: they evaluate on diversity-focused or redundancy-heavy benchmarks---AdaGReS targets near-duplicate chunks by design---or with older first-stage encoders such as Contriever \cite{izacard2022unsupervised}. None isolates the diversification step on pools drawn from nonredundant corpora with state-of-the-art encoders, the configuration of most production vector-search deployments.

	\subsection{The case against diversity in RAG}
	\label{sec:contradiv}

	A complementary literature reaches the opposite conclusion. For web search, \citet{santos2012novelty} showed that novelty-based diversification, which promotes documents merely for being dissimilar to those already selected, does not consistently improve over a standard non-diversified ranking and distinguished such geometric novelty from genuine coverage of query aspects. Their analysis, however, predates RAG and does not relate the value of diversification to measured candidate-pool redundancy or to the evidence requirements of the query, the axes studied here. Controlled evaluations find that \MMR offers no significant advantage over naive nearest-neighbor retrieval \cite{eibich2024aragog}, that reranking is among the highest-leverage components of a RAG pipeline \cite{wang2024bestpractices} and, in a biomedical study holding the generator and embeddings fixed, that \MMR ``sacrifices answer relevancy for diversity'' while the dense baseline performs comparably to the best configuration \cite{bal2026biomedrag}. Studies of context composition show that passages semantically related to the query but not containing the answer are precisely the ones that harm the generator \cite{cuconasu2024noise,amiraz2025distracting} and geometric diversification, by design, substitutes exactly such passages. Analyses of context sufficiency \cite{joren2025sufficient} likewise indicate that generation needs coverage of the required information, which geometric spread does not measure.

	\subsection{Diversification in recommendation}
	\label{sec:recsys-related}

	Recommender systems are the setting in which diversity is given intrinsic evaluation credit through beyond-accuracy metrics that reward explicitly genre coverage and intra-list dissimilarity \cite{vargas2011rank,kaminskas2017diversity}. The standard methodology is post-processing reranking: a relevance-optimized recommender produces a size-$m$ candidate list per user and a greedy reranker selects $n<m$ items balancing relevance against diversity \cite{carbonell1998mmr,wilhelm2018practical,carraro2024llmdiversity}. Controlled comparisons under this protocol consistently find that diversification pays with \DPP delivering strong coverage gains at deployment scale \cite{wilhelm2018practical,chen2018fastdpp}. Search-result diversification reaches the same conclusion under intent-aware metrics \cite{santos2010exploiting,agrawal2009diversifying,clarke2008alphaNDCG}.

	When diversity is given intrinsic credit, the central question of this paper, namely whether diversification in the retrieval stage of a RAG pipeline helps or harms, does not arise, since diversity is then rewarded in its own right. We therefore scope the study to instrumental credit regimes, where diversity can pay only through the final answer, and refer the reader to the literature above for the intrinsic credit regime.

	\subsection{Diversity evaluation metrics}
	\label{sec:divmetrics}

	Separating useful coverage from raw spread requires both families of diversity metrics. Annotation-free measures, average pairwise distance (APD) and the Vendi score \cite{friedman2023vendi}, quantify geometric spread without reference to task structure. Intent-aware measures from the diversification literature, $\alpha$-NDCG \cite{clarke2008alphaNDCG}, ERR-IA \cite{chapelle2011errIA} and Subtopic Recall \cite{zhai2015subtopicrecall}, reward covering the distinct facets a query requires and penalize redundant coverage. We derive subtopics without extra annotation: for multi-hop QA, the distinct gold supporting documents of a question; for retrieval tasks with relevance judgements, the relevant-document set of the query. Disagreement between the two families is diagnostic: a method that raises APD while lowering $\alpha$-NDCG produces spread without coverage. The same annotation-free statistics, computed on the embedded pool before any labels are consulted, also serve as the redundancy diagnostics of Section~\ref{sec:diagnostics}.

	\subsection{Relative neighborhood graphs}
	\label{sec:rng-related}

	The \emph{relative neighborhood graph} (\RNG) was originally defined for finite planar point sets under the Euclidean distance \cite{toussaint1980rng} and extends naturally to any finite metric space \cite{jaromczyk1992rng}. The soft score studied here requires still less structure, namely a symmetric nonnegative dissimilarity (Appendix~\ref{app:ce-variants}). We say two points $x,y$ are relative neighbors when no third point $z$ satisfies both $d(z,x)<d(x,y)$ and $d(z,y)<d(x,y)$. Equivalently, the open lune of $x$ and $y$, namely the intersection of the open balls of radius $d(x,y)$ centered at $x$ and $y$, is empty (Figure~\ref{fig:lune-rng}).

	\begin{figure}[ht]
		\centering
		\begin{subfigure}[b]{0.46\linewidth}
			\centering
			\begin{tikzpicture}[scale=0.85,
				pt/.style={circle,draw,fill=black,inner sep=2pt}]
				\clip (-1.2,-2.8) rectangle (4.2,2.8);
				\begin{scope}
					\clip (0,0) circle (3);
					\fill[figblue!12] (3,0) circle (3);
				\end{scope}
				\draw[figgray] (0,0) circle (3);
				\draw[figgray] (3,0) circle (3);
				\draw[dashed, thin, figgray] (0,0) -- (3,0);
				\node[tinylabel] at (1.5,0.22) {$d(x,y)$};
				\node[pt, label={[tinylabel]below left:$x$}] at (0,0) {};
				\node[pt, label={[tinylabel]below right:$y$}] at (3,0) {};
				\node[activenode, label={[tinylabel]above right:$z_1$}] at (1.5,1.2) {};
				\node[cand, label={[tinylabel]above:$z_2$}] at (0.5,-2.2) {};
			\end{tikzpicture}
			\subcaption{The lune of $x$ and $y$.}
			\label{fig:lune}
		\end{subfigure}\hfill
		\begin{subfigure}[b]{0.50\linewidth}
			\centering
			\begin{tikzpicture}[scale=1.05,
				pt/.style={circle,draw,fill=black,inner sep=2pt}]
				  \draw[thick, figblue] (3.08,3.07) -- (3.06,3.36);
  \draw[thick, figblue] (3.08,3.07) -- (2.34,3.13);
  \draw[thick, figblue] (3.77,0.92) -- (3.87,1.70);
  \draw[thick, figblue] (3.77,0.92) -- (3.02,0.34);
  \draw[thick, figblue] (1.58,3.00) -- (0.22,2.83);
  \draw[thick, figblue] (1.58,3.00) -- (2.34,3.13);
  \draw[thick, figblue] (0.22,2.83) -- (0.36,1.85);
  \draw[thick, figblue] (3.87,1.70) -- (2.52,1.97);
  \draw[thick, figblue] (3.87,1.70) -- (4.78,2.74);
  \draw[thick, figblue] (1.59,1.09) -- (1.37,1.62);
  \draw[thick, figblue] (1.59,1.09) -- (1.19,0.71);
  \draw[thick, figblue] (1.59,1.09) -- (3.02,0.34);
  \draw[thick, figblue] (1.37,1.62) -- (2.52,1.97);
  \draw[thick, figblue] (1.37,1.62) -- (0.36,1.85);
  \draw[thick, figblue] (2.52,1.97) -- (2.34,3.13);
  \node[pt] at (3.08,3.07) {};
  \node[pt] at (3.77,0.92) {};
  \node[pt] at (1.58,3.00) {};
  \node[pt] at (0.22,2.83) {};
  \node[pt] at (3.87,1.70) {};
  \node[pt] at (1.59,1.09) {};
  \node[pt] at (1.37,1.62) {};
  \node[pt] at (2.52,1.97) {};
  \node[pt] at (4.78,2.74) {};
  \node[pt] at (3.06,3.36) {};
  \node[pt] at (1.19,0.71) {};
  \node[pt] at (3.02,0.34) {};
  \node[pt] at (0.36,1.85) {};
  \node[pt] at (2.34,3.13) {};
			\end{tikzpicture}
			\subcaption{A relative neighborhood graph.}
			\label{fig:rng-graph}
		\end{subfigure}
		\caption{(\subref{fig:lune}) The open lune of $x$ and $y$ (shaded) is the
			intersection of the two open balls of radius $d(x,y)$ centered at $x$ and
			$y$; $z_1$ lies inside it and obstructs the pair $\{x,y\}$, so that pair is
			not an edge of the \RNG, while $z_2$ does not. (\subref{fig:rng-graph}) The
			\RNG\ of a planar point set, with an edge between exactly the pairs whose
			lune is empty.}
		\label{fig:lune-rng}
	\end{figure}
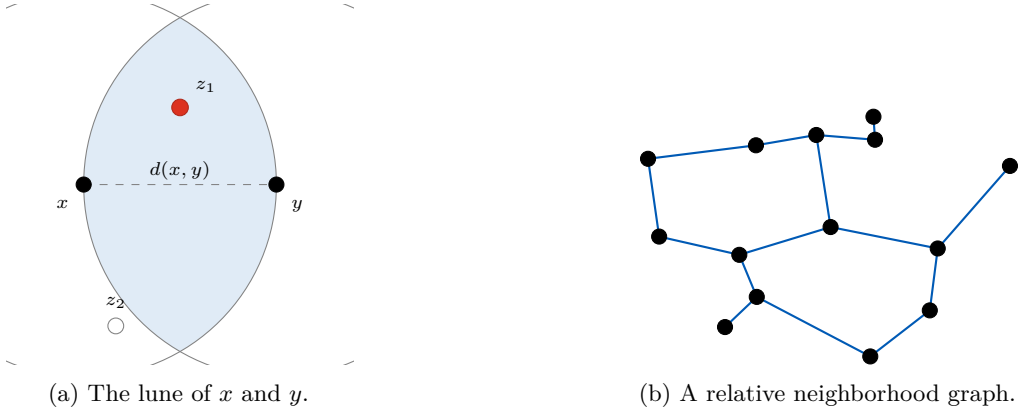

	In the retrieval setting, $x=q$ is the query and $y=w$ is a candidate document. A third document $v$ with $d(q,v)<d(q,w)$ and $d(v,w)<d(q,w)$ lies in the lune of $q$ and $w$: it is closer to the query than $w$ and also close to $w$, making $w$ geometrically redundant relative to $q$. Such a point $v$ is said to \emph{obstruct} the pair $(q,w)$.

	\RNG-style edge pruning is foundational in graph-based approximate nearest-neighbor indexing: HNSW, NSG and DiskANN sparsify neighborhood graphs by discarding edges whose lune contains another indexed point, at index-construction time and between corpus points, with DiskANN's $\alpha$-pruning relaxing the lune test by a multiplicative slack \cite{malkov2018hnsw,fu2019nsg,subramanya2019diskann}. Query-centered geometric diverse-neighbor retrieval itself also predates the present work: \citet{kucuktunc2013lambda} introduced $\lambda$-diverse nearest-neighbor browsing for multidimensional data, using proximity-graph constructions to return neighbors that are simultaneously close to the query and mutually dispersed. Recent work instead brings diversity into \ANN\ retrieval, either by diversifying the graph search itself \cite{weber2024hnsw} or by designing data structures for diverse similarity queries under a hard diversity objective \cite{anand2025diverseann}.

	The present paper uses the lune condition at query time, anchored at the query $q$ as in Figure~\ref{fig:knn-vs-rng}, but softens the binary test into a cumulative penalty that ranks a fixed pool rather than pruning a graph or enforcing a hard constraint. The contribution is the cumulative relative-neighborhood hinge score of Definition~\ref{def:rngscore}, with its explicit nearest-neighbor limit analysis, together with an evidence-aware selection rule, evaluated across controlled RAG redundancy regimes, none of which appears in these earlier constructions.

	\section{The Relative-Neighborhood Score}
	\label{sec:method}

	In order to analyze the impact of diversification at the retrieval stage of a RAG pipeline, we use different diversifiers such as \MMR \cite{carbonell1998mmr} and \DPP \cite{kulesza2012dpp,chen2018fastdpp}. Here, we introduce the \RNGSCORE, a soft score used as a reranker and which can be tuned either to diversify retrieval or to fall back to $k$-NN as needed.

	\subsection{Setup}

	Let $V=\{v_1,\dots,v_m\}\subset \R^D$ be a finite candidate pool for a fixed query $q\in\R^D$. We assume the pool is the result of a first-stage retriever and that $m \gg k$ but remains moderate enough for pairwise reranking. Distances are denoted by $d(\cdot,\cdot)$. Throughout the paper, it is enough to assume that $d$ is a \emph{symmetric nonnegative dissimilarity} on $V\cup\{q\}$ (see Appendix \ref{app:ce-variants} for the precise definition). Candidates are ranked by increasing score, with ties broken by the deterministic sorting routine of the implementation applied identically by every method including $k$-NN.

	For a candidate $w\in V$, define its set of \emph{closer competitors}
	\[
	P(w;q,V)=\{v\in V : d(q,v) < d(q,w)\}.
	\]
	These are the only documents that can make $w$ redundant relative to the query. Because the inequality is strict, candidates at exactly equal query distance never obstruct one another: exactly coincident duplicates are each scored as if the other were absent and ranked as $k$-NN ranks them (both can be selected), so the suppression the score provides targets near-duplicates rather than verbatim copies (Section~\ref{sec:theory}). Score ties fall under the deterministic rule above: since tied candidates are in practice verbatim copies sharing identity, text and embedding, no reported value depends on their resolution.

	\subsection{Definition}

	\begin{defn}[\RNGSCORE]
		\label{def:rngscore}
		Let $\gamma\in\R$. For $w\in V$, define
		\begin{equation}
			\label{eq:rngscore}
			\score{\gamma}(w;q,V)
			= d(q,w) + \sum_{v\in P(w;q,V)} \hinge{d(q,w)-d(v,w)+\gamma},
		\end{equation}
		where $\hinge{x}=\max(x,0)$.
	\end{defn}

	The first term is simply the distance from $w$ to the query, while each summand is an \emph{obstruction penalty}. It is positive exactly when $d(v,w) < d(q,w)+\gamma$, i.e. when a closer document $v$ sits close enough to $w$ to make $w$ look redundant from the viewpoint of $q$. Positive $\gamma$ makes the penalty stricter, negative $\gamma$ relaxes it. The induced reranker keeps the $k$ lowest-scoring candidates with a computational cost of $O(m^{2}D)$: measured at $m=100$, $k=5$ and $D=1024$, the full reranking pass costs $230\,\mu$s per query on average (median $214\,\mu$s, $95$th percentile $378\,\mu$s),\footnote{\label{fn:latency-protocol}Single-query CPU timings on one pinned $3.9$\,GHz Cortex-X925 core with one OpenBLAS thread (NumPy 2.4.6): $100$ untimed warm-up iterations followed by $10^{4}$ timed calls at batch size one. The timing harness and full latency distributions are released with the code.} negligible against embedding and generation, so reranking cost is not studied further.

	\begin{figure}[t]
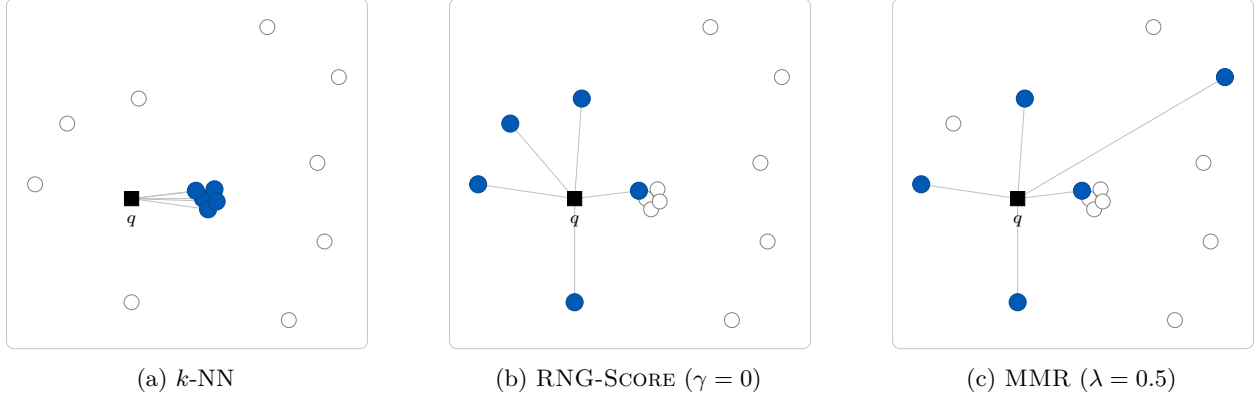

		\centering
		\begin{subfigure}[b]{0.29\linewidth}
			\centering
			\resizebox{\linewidth}{!}{\begin{tikzpicture}\draw[figgray!45, rounded corners=3pt, thin] (-1.75,-2.1) rectangle (3.3,2.8);\input{plots/fig_sel_knn.tex}\end{tikzpicture}}
			\subcaption{$k$-NN}
			\label{fig:sel-knn}
		\end{subfigure}\hfill
		\begin{subfigure}[b]{0.29\linewidth}
			\centering
			\resizebox{\linewidth}{!}{\begin{tikzpicture}\draw[figgray!45, rounded corners=3pt, thin] (-1.75,-2.1) rectangle (3.3,2.8);\input{plots/fig_sel_rng.tex}\end{tikzpicture}}
			\subcaption{\RNGSCORE ($\gamma=0$)}
			\label{fig:sel-rng}
		\end{subfigure}\hfill
		\begin{subfigure}[b]{0.29\linewidth}
			\centering
			\resizebox{\linewidth}{!}{\begin{tikzpicture}\draw[figgray!45, rounded corners=3pt, thin] (-1.75,-2.1) rectangle (3.3,2.8);\input{plots/fig_sel_mmr.tex}\end{tikzpicture}}
			\subcaption{\MMR\ ($\lambda=0.5$)}
			\label{fig:sel-mmr}
		\end{subfigure}
		\caption{One candidate pool, three computed top-$5$ selections (filled blue
			= selected, open = unselected, $\blacksquare = q$).
			(\subref{fig:sel-knn}) $k$-NN spends the whole budget on a tight cluster of
			near-duplicates; (\subref{fig:sel-rng}) the \RNGSCORE keeps one cluster
			representative and disperses to nearby distinct points; (\subref{fig:sel-mmr}) \MMR pushes further out onto less relevant points.}
		\label{fig:knn-vs-rng}
	\end{figure}

	Table~\ref{tab:example} shows the three selectors on a real HotpotQA fullwiki query: $k$-NN fills the list with near-identical sibling articles and misses the second gold document, the \RNGSCORE keeps one cluster representative and recovers the second document while \MMR disperses onto less relevant material.

	\begin{table}[ht]
		\centering
		\small
		\caption{A HotpotQA fullwiki query needing two pieces of evidence
			(\emph{1925 Birthday Honours} and \emph{George V}, in bold).
			}
		\label{tab:example}
		\begin{tabular}{rlll}
			\hline
			\textbf{rank} & \textbf{$k$-NN} & \textbf{\RNGSCORE ($\gamma{=}0$)} & \textbf{\MMR\ ($\lambda{=}0.5$)} \\
			\hline
			1 & \textbf{1925 Birthday Honours} & \textbf{1925 Birthday Honours} & \textbf{1925 Birthday Honours} \\
			2 & 1915 Birthday Honours  & David Lloyd George & Edward VII \\
			3 & 1924 Birthday Honours  & 1915 Birthday Honours & George Champagne \\
			4 & 1926 Birthday Honours  & Henry V, Holy Roman Emperor & Julius Meinl V \\
			5 & 1927 Birthday Honours  & \textbf{George V} & A.\ Thynn, Marquess of Bath \\
			\hline
		\end{tabular}
	\end{table}

	\subsection{Properties and instantiation}
	\label{sec:theory}

	The score deviates from pure relevance only where there is geometric obstruction and is otherwise pinned to nearest-neighbor ranking, which is what lets its tuned parameter reflect the pool rather than the tuner. In particular, it reduces to $k$-NN ranking at both extremes of the margin (Propositions~\ref{prop:knnlimit} and~\ref{prop:knnlimit_pos}), so a margin grid containing a deactivating margin ($\gamma\le\gamma_0$ in Proposition~\ref{prop:knnlimit}, under cosine distance $\gamma=-2$ always suffices) contains the relevance-only baseline as a reachable operating point (Corollary~\ref{cor:fallback}). Validated tuning over such a grid never underperforms $k$-NN on the validation set itself: this is an optimization identity and Theorem~\ref{thm:valbound} gives the corresponding finite-sample statement for unseen queries. The evaluated grids stop short of $\gamma=-2$, but that margin is $k$-NN itself, which every table reports, so the theorem applies without an additional run. The fallback carries no guarantee at interior margins: since the score consults no relevance labels, an interior $\gamma$ can reorder the selection adversely for the task, so worst-case behavior over the grid is an empirical quantity, reported in Section~\ref{sec:rq2-regret}. \MMR\ shares a one-sided version of the fallback, recovering pure relevance at $\lambda=1$, whereas Maxmin and \DPP have no operating point that recovers relevance-only ranking. The \RNGSCORE\ is distinctive in reaching the fallback from both margin extremes, which keeps the tuned margin interpretable as a pool diagnostic throughout the sweep.

	The obstruction mechanism also delimits the score's scope along the duplication axis. The hinge penalizes a candidate through strictly closer near-copies: for $v\in P(w;q,V)$ with $d(v,w)\to0$, the penalty approaches $[\,d(q,w)+\gamma\,]_+$, its largest attainable value for that pair, so near-duplicates of better-ranked content are suppressed most strongly. Exactly coincident copies, by contrast, are not competitors of one another under the strict inequality and are not suppressed at all: on byte-exact duplication the score reduces to $k$-NN behavior and its redundancy gains in the sweeps come entirely from the perturbed copies. The score therefore targets the semantic near-duplicate structure studied in this paper, while byte-exact duplication is the regime of hash-level deduplication \cite{schelpe2026dedup}, removable upstream at negligible cost.

	\label{sec:kernel}
	In many RAG pipelines, documents in the corpus are normalized and similarity between documents $v$ and $w$ is measured using cosine similarity, which we note $\langle v,w\rangle$. In this context, the corresponding distance is $d(v,w) = 1 - \langle v,w\rangle$ and the \RNGSCORE for given query $q$ and pool $V$ becomes
	\begin{equation}
		\label{eq:cosine-score}
		\score{\gamma}^{c}(w;q,V)
		= 1-\langle q,w\rangle
		+\sum_{\substack{v\in V\\ \langle q,v\rangle>\langle q,w\rangle}}
		\hinge{\langle v,w\rangle-\langle q,w\rangle+\gamma}.
	\end{equation}
	Furthermore, the \RNGSCORE is adaptable to the context where the query--document distance differs from the distance in the embedding space, as is the case when using a cross-encoder (Appendix~\ref{app:ce-variants}). We may also adapt the \RNGSCORE so that it takes as input a kernel analogously to \DPP (Appendix~\ref{app:kernel}).

	\section{Redundancy Diagnostics and the Decision Rule}
	\label{sec:diagnostics}

	The regime analysis of this paper is only actionable if the controlling variable can be observed by a deployed system cheaply and without labels. This section defines the two diagnostics we use, at the deployment level and at the query level, and the decision rule built on the latter. Both diagnostics are computed from the embedded candidate pool that the retrieval system already holds: neither consults relevance judgements, gold answers or any model beyond the existing encoder.

	\subsection{Near-duplicate pair fraction}
	\label{sec:pool-redundancy}

	For an embedded pool $V$ with unit-normalized embeddings, define
	\begin{equation}
		\label{eq:pool-redundancy}
		\operatorname{PoolRedundancy}(V)
		= \binom{|V|}{2}^{-1}
		\left|\big\lbrace\,\{v,w\}\subset V : \langle v,w\rangle > 0.95 \,\big\rbrace\right|,
	\end{equation}
	the fraction of passage pairs whose cosine similarity exceeds a near-duplicate threshold, fixed at $0.95$ throughout. Two properties of the statistic should be kept in mind when comparing across deployments: the cutoff's meaning depends on the encoder's similarity scale and the pair fraction grows quadratically with duplicate-cluster size. Thus, comparisons should hold the encoder and the pool size $m$ fixed, as all sweeps in this paper do. In our experiments $V$ is each query's top-$m$ candidate pool and we report per-condition means over queries. Averaged over a sample of production queries, the same statistic places any deployment on the common measured-redundancy axis along which we report the synthetic injection sweep and the natural chunking sweep of Section~\ref{sec:rq2}, so that a measured value can be compared directly against the crossover observed in our experiments.

	\subsection{Effective size of the nearest-neighbor selection}
	\label{sec:trigger}

	The pair fraction \eqref{eq:pool-redundancy} is a property of the whole pool. Nevertheless, what determines whether diversification helps for a specific query is whether redundancy crowds the top of that query's ranking. We therefore use a sharper, selection-centered statistic. We run plain nearest-neighbor selection and measure the \emph{effective number of distinct documents} among the $k$ selected, the Vendi score \cite{friedman2023vendi} of the embedding similarity matrix of the selection:
	\begin{equation}
		\label{eq:trigger}
		T(q) = \operatorname{Vendi}\bigl(k\text{-}\mathrm{NN}(q)\bigr)
		\;\in\; [1, k].
	\end{equation}
	$T(q)\approx k$ means the $k$ selected documents are mutually orthogonal, while $T(q)\approx 1$ means the selection has collapsed onto near-copies of a single document. Following the semantics of the Vendi score \cite{friedman2023vendi}, the effective number of distinct documents measures how many geometrically distinct embeddings the selection effectively contains, not how many semantically distinct evidence pieces it covers; how well this geometric quantity serves as a per-query regime diagnostic is answered empirically in Section~\ref{sec:rq4}. The statistic costs one $k\times k$ eigendecomposition on embeddings already in memory: under the measurement protocol of footnote~\ref{fn:latency-protocol}, the full trigger, kernel construction included, costs $20.5\,\mu$s per query on average at $k=5$ and $D=1024$ (median $20.3\,\mu$s, $95$th percentile $25.8\,\mu$s).

	\subsection{The decision rule}
	\label{sec:rule}

	The decision rule is a one-line policy with two frozen choices, a threshold $\tau\in[1,k]$ and a fallback diversifier $D$:
	\begin{equation}
		\label{eq:rule}
		\text{select}(q) =
		\begin{cases}
			D(q) & \text{if } T(q) < \tau,\\[2pt]
			k\text{-}\mathrm{NN}(q) & \text{otherwise.}
		\end{cases}
	\end{equation}
	The threshold is set to the budgeted evidence requirement, $\tau=\min\{h,k\}$, where $h$ is the number of distinct evidence pieces a query needs: since $T(q)\in[1,k]$, the cap keeps the threshold inside the range of the statistic. Two supporting documents give $h=2$ on the multi-hop benchmarks, while $h=1$ makes the trigger condition $T(q)<1$ unsatisfiable, so on single-hop retrieval the rule reduces to $k$-NN by construction. Evidence for the diagnostic itself accordingly comes from the interior cases between this endpoint and the saturated one where the rule nearly always triggers. Where $h$ is read from gold annotations (hop counts, relevant-set sizes), we label the results oracle analyses of the rule under a known evidence requirement, the natural upper reference for a deployed policy that estimates $h$ from its workload without consulting test labels. On the graded-relevance tasks, relevant-document count serves as a proxy for evidence multiplicity, since several relevant documents can be substitutes for one fact rather than complementary facts. The fallback diversifier $D$ is selected on a validation set pooled across redundancy conditions and then frozen.

	The rule is a realizable member of the family of query-adaptive policies that the all-methods oracle of Section~\ref{sec:rq3} bounds from above, which makes the comparison meaningful in both directions: the oracle says how much any policy routing among the evaluated methods could gain, while the rule says how much of that a single-threshold policy actually captures. Section~\ref{sec:rq4} confirms the anchoring empirically: threshold tuning left free recovers $\tau^{*}\approx h$ at both $k=5$ and $k=10$, with a wide plateau of equivalent thresholds around it.

	Note that the paper thus carries two thresholds with different roles. The near-duplicate pair fraction is a statistic measured on a pool of $m$ documents while the rule threshold $\tau$ lives on the individual query: it tests the selection of $k\ll m$ documents the system has just made. The two are linked through the rule trigger rate: as a pool's measured redundancy passes a certain crossover (see for example Figure~\ref{fig:regimes-summary}), the per-query statistic $T(q)$ collapses for most queries and the fraction of queries falling below $\tau$ jumps from under $16\%$ to over $80\%$ (Section~\ref{sec:rq4}). Thus, thresholding each query at $\tau$ implements the pool-level boundary automatically without ever measuring the pool redundancy. The validation-tuned \RNGSCORE margin of Section~\ref{sec:method} supplies a third complementary readout, offline and at the workload level since validation queries are needed to compute it.

	\section{Experimental Setup}
	\label{sec:setup}

	This section details the research questions, datasets, models, baselines, evaluation metrics and protocols. The results follow in Section~\ref{sec:results}.

	\subsection{Research questions}
	\label{sec:rqs}

	The evaluation is organized around five research questions.
	\begin{enumerate}[label=\textbf{RQ\arabic{enumi}}]
		\item\label{rq:benchmarks} \emph{Diversification on standard benchmarks.} On pools retrieved from nonredundant corpora with state-of-the-art encoders, does any geometric diversifier improve results quality over nearest-neighbor selection?
		\item\label{rq:redundancy} \emph{Manipulating pool redundancy.} Does raising pool redundancy create a regime in which diversification pays and what is the risk of committing to a method without knowing the regime?
		\item\label{rq:oracle} \emph{Oracle headroom.} How much could a query-adaptive policy routing among the evaluated diversifiers gain in each regime?
		\item\label{rq:rule} \emph{The decision rule.} Can the diagnostic of Section~\ref{sec:trigger} decide per query whether to diversify and how much of the oracle headroom of \ref{rq:oracle} does the rule capture?
		\item\label{rq:answerquality} \emph{Answer quality under redundancy.} Do the findings carry over to end-to-end generation quality in a redundant regime?
	\end{enumerate}

	\subsection{Datasets}
	\label{sec:datasets}

	\emph{HotpotQA} (fullwiki) \cite{yang2018hotpotqa} is the primary multi-hop benchmark: each question needs two supporting paragraphs from distinct articles and candidate pools are the unique passages of the top-100 MDR chains of \citet{xiong2021mdr} (between 80--120 candidates per query, two gold passages). \emph{2WikiMultiHopQA} \cite{ho2020wikimultihop} and \emph{MuSiQue} \cite{trivedi2022musique} (two to four hops) vary hop count at fixed credit and redundancy regimes. Candidate pools for 2WikiMultiHopQA are the ten paragraphs of its distractor setting, the supporting paragraphs plus distractors. MuSiQue, restricted to its answerable subset, provides roughly twenty paragraphs per question (two to four of them supporting), which form its pools. \emph{Natural Questions Open} \cite{kwiatkowski2019nq} with DPR top-100 pools \cite{karpukhin2020dpr} is a single-hop, single-answer control. For the redundancy sweep, we collapse a query's answer-containing passages to one subtopic, since DPR marks every such passage relevant (around $4.6$ per query) and would otherwise misread as multi-evidence. \emph{BEIR} \cite{thakur2021beir} supplies heterogeneous graded-relevance tasks (SciFact, FiQA-2018 and TREC-COVID together with the argument-retrieval tasks ArguAna and Webis-Touch\'e-2020), with each query's relevant set as its subtopics. SciFact anchors the single-evidence end ($92\%$ of queries one relevant document), FiQA-2018 the multi-evidence middle (typically two to three relevant documents per query) and TREC-COVID the saturated extreme (mean $493$).

	\subsection{Models}
	\label{sec:models}

	Because the analysis applies to any finite candidate pool, we evaluate along an encoder-capacity axis. The strong stack uses \texttt{bge-m3} \cite{Chen2024} (1024-dimensional) and \texttt{Qwen3-Embedding-4B} \cite{Zhang2025Qwen3} (2560-dimensional) as dense encoders and \texttt{bge-reranker-v2-m3}, the reranker released alongside BGE-M3 \cite{Chen2024}, as the cross-encoder for the multi-stage half of \ref{rq:benchmarks}. The weak encoder \texttt{all-MiniLM-L6-v2} \cite{reimers2019sbert}, representative of lightweight deployments, supplies a capacity contrast for \ref{rq:redundancy}. Encoder capacity is treated as a third route into the redundant regime, not as a robustness check. Downstream generation uses \texttt{flan-t5-base} \cite{chung2022scaling} as the reader (4-beam decoding); the EM/F1 columns of Table~\ref{tab:rq1-hotpot} and the clean- and heavy-redundancy comparison of Section~\ref{sec:rq5} are instead computed with the instruction-tuned \texttt{Qwen3.8-27B} \cite{qwenteam2026qwen38} under deterministic decoding.\footnote{Every \texttt{Qwen3.8-27B} run uses one fixed prompt: the instruction ``Answer the question based on the context. Give only a short, direct answer.'' followed by the question and the $k$ selected passages, title-prefixed and presented in retrieval order, rendered through the model's native chat template with reasoning disabled. Decoding is deterministic: temperature $0$, top-$p$ $1$, seed $0$, at most $64$ new tokens, no beam search, within an $8{,}192$-token context. The model runs in BF16 at a pinned weight revision; the serving configuration is recorded with the released code.}

	\subsection{Baselines}
	\label{sec:baselines}

	We compare $k$-nearest-neighbor selection ($k$-NN) against four reranking baselines applied to the same candidate pool: \MMR \cite{carbonell1998mmr} with trade-off weight $\lambda\in\{0.3,0.5,0.7\}$, Maxmin dispersion \cite{ravi1994maxmin}, greedy \DPP selection \cite{kulesza2012dpp,chen2018fastdpp} and the \RNGSCORE with margin grid $\gamma\in\{-0.5,-0.3,-0.1,-0.05,0,0.05,0.1,0.3,0.5\}$.\footnote{These are the default grids: the redundancy and chunking sweeps refine both and reach further at their upper ends (Section~\ref{sec:protocol}), and the single-run analyses (the non-HotpotQA cross-encoder comparisons of Table~\ref{tab:rq2-ce} and the ArguAna and Webis-Touch\'e columns of Table~\ref{tab:rq1-beir}) use a fine margin grid, $\gamma\in[-1,1.25]$ in steps of $0.05$.} Maxmin is query-agnostic and is included as a diversity stress-test rather than as a model of relevance-aware reranking. Greedy \DPP\ uses the quality-weighted kernel $L=\operatorname{diag}(s)\,C\,\operatorname{diag}(s)$ ($C$ the cosine Gram matrix of the pool, $s$ the query similarities clipped at zero) with greedy MAP inference by incremental orthogonalization \cite{chen2018fastdpp}. We refer to \MMR, Maxmin and greedy \DPP\ collectively as \emph{fixed-pressure} diversifiers: each applies a diversification pressure set by its parameters, with no operating point in the evaluated grids that recovers relevance-only ranking. \MMR\ does recover it at the endpoint $\lambda=1$, where its selection criterion reduces to relevance alone and reproduces the $k$-NN ranking exactly under the shared tie rule. The $k$-NN rows of the tables therefore already report this endpoint and the grids list only interior operating points, so the tuned \MMR\ of the sweeps is the best \emph{interior} one and its clean-pool toll and sign change (Section~\ref{sec:rq2}) concern interior diversification; an \MMR\ tuned over a grid containing $\lambda=1$ would choose between that row and $k$-NN on validation, a policy that the per-level tuned row of Table~\ref{tab:rq6-rule} bounds on validation. Within the evaluated grids, where $\lambda=1$ is represented by the $k$-NN row, the \RNGSCORE is the only method whose tuning can express ``do not diversify''.

	\subsection{Evaluation metrics}
	\label{sec:metrics-eval}

	We report four metric families so that relevance, raw spread, genuine coverage and downstream utility can be distinguished. \emph{Relevance}: Recall@$k$ and NDCG@$k$ against gold passages or graded judgements. \emph{Annotation-free diversity}: average pairwise cosine distance (APD) and the Vendi score \cite{friedman2023vendi}. \emph{Intent-aware diversity}: $\alpha$-NDCG@$k$ \cite{clarke2008alphaNDCG} with redundancy parameter $\alpha=0.5$ and Subtopic Recall (S-Recall@$k$) \cite{zhai2015subtopicrecall}, with subtopics derived as in Section~\ref{sec:datasets}. \emph{Generation quality}: exact match (EM) and token-level F1 against reference answers. All relevance, intent-aware and generation metrics are reported on a $0$--$100$ scale, whereas APD and the Vendi score are reported on their natural scales. Recall@$k$ counts all gold passages of a query, including the roughly $5\%$ that the first-stage pool fails to contain while S-Recall@$k$ counts retrievable subtopics only, which is why the two differ at the baseline.

	\emph{Misspecification downside.} If the value of diversification depends on a regime variable that practitioners do not usually measure, then the practically relevant property of a method is not its tuned performance in either regime but its worst case across regimes. Writing $U_r(M_\theta)$ for the test utility (S-Recall@5) of method $M$ at operating point $\theta$ in regime $r$, we report
	\begin{equation}
		\label{eq:regret}
		\operatorname{downside}(M)
		\;=\;
		\max_{r}\;\max_{\theta}\;
		\bigl[\,U_r(k\text{-}\mathrm{NN}) - U_r(M_\theta)\,\bigr]_{+},
	\end{equation}
	the worst shortfall against the relevance-only default, over distinct redundancy regimes and over the method's full hyperparameter grid (the sweep grids of Section~\ref{sec:protocol}). The quantity is a maximum over the grid, so the values computed on the sweep grids upper-bound those of the narrower default grids (whose only member absent from the sweep grids, $\gamma=-0.5$, reproduces $k$-NN) and it is intrinsically grid-dependent: enlarging a grid with a poor operating point increases it. For $k$-NN, which commits no operating point, the reference is instead the best tuned method in each regime: that row is therefore constructed differently from the others.

	\emph{Minimax regret.} The downside is anchored at the $k$-NN default. Its standard decision-theoretic complement \cite{savage1951decision} measures every method against the same per-regime oracle. For the common action library $\mathcal{C}$ consisting of $k$-NN and every evaluated operating point of the compared methods, let $U_r^{*}=\max_{a\in\mathcal{C}}U_r(a)$ be the oracle utility in regime $r$. A method that must commit one operating point $\theta$ from its grid $\Theta_M$ before the regime is known has minimax regret
	\begin{equation}
		\label{eq:minimax}
		\operatorname{regret}(M)
		\;=\;
		\min_{\theta\in\Theta_M}\;\max_{r}\;
		\bigl\{\,U_r^{*} - U_r(M_\theta)\,\bigr\},
	\end{equation}
	with $\operatorname{regret}(k\text{-}\mathrm{NN})=\max_r\{U_r^{*}-U_r(k\text{-}\mathrm{NN})\}$ for the parameter-free default. Here, every row uses the same $U_r^{*}$.

	The two quantities answer different questions, namely the worst shortfall against the $k$-NN default a practitioner would otherwise run versus the worst distance from the best fixed commitment in hindsight. Both are grid-dependent and reported in Table~\ref{tab:rq3-regret}.

	\subsection{Protocol}
	\label{sec:protocol}

	For each query, a first-stage retriever produces a candidate pool of size $m$ and every reranking method selects $k$ documents from the same pool, so that all differences are attributable to the selection step. Unless stated otherwise, $m=100$ and $k=5$ for QA tasks while $m=100$ and $k=10$ for BEIR tasks. Hyperparameters ($\gamma$ for the relative-neighborhood scores, $\lambda$ for \MMR) are selected on a held-out validation split of $20\%$ of the queries and reported on the remaining test split.\footnote{In practice the negative end of the default margin grid realizes the fallback of Corollary~\ref{cor:fallback}: at $\gamma=-0.5$ the selected sets are identical to $k$-NN's on every query of the seed-replicated BEIR tasks and match $k$-NN to reporting precision on the clean HotpotQA pools, so validated tuning can express ``do not diversify''. The sweep grid's $-0.3$ end deviates from $k$-NN on about $1\%$ of clean queries.} Unless identified as single-run, every sweep and comparison table is repeated under three seeds, the seed controlling both the injected near-duplicates and the validation/test split. On HotpotQA, each split is $1{,}481$ validation and $5{,}924$ test queries. The same frozen splits underlie \ref{rq:benchmarks} and \ref{rq:redundancy}--\ref{rq:rule}, so that all derived analyses are computed from one set of per-query results. Sweeps execute every query, but hyperparameters are selected on the validation split alone and every reported summary and inferential statement is computed exclusively on the per-seed test split.

	\paragraph{Statistical reporting.}
	 Unless stated otherwise, every reported value is a mean over the three seeds, and every comparison is a paired difference against $k$-NN on the same pools, evaluated on test queries only. Uncertainty is quantified with query-level bootstrap confidence intervals, and ``significant'' means a multiplicity-adjusted $p$-value below $0.05$. Quoted paired effects and intervals are computed from unrounded per-query means, in some cases under a query-level aggregation where the corresponding table displays seed means, so the difference of two displayed table cells can deviate from the quoted paired effect in the last displayed digit. The primary test of each sweep asks whether a method's paired effect improves from the clean to the most redundant condition. Effects within one point on the $0$--$100$ metric scales are treated as immaterial regardless of significance, while a nonsignificant difference is never read as evidence of equivalence. Single-run and aggregate-only results are identified as such and reported descriptively. The complete procedure, its provenance and its limits are given in Appendix~\ref{app:stats}.

	\paragraph{Cross-encoder pipeline.}
	 The multi-stage experiments retrieve the top-$m$ passages with \texttt{bge-m3}, rescore them with \texttt{bge-reranker-v2-m3} and hand the rescored pool to the rerankers. Because the \RNGSCORE is defined through a distance while a cross-encoder returns only query--document relevance, these experiments use the dissimilarity construction of Appendix~\ref{app:ce-variants}, with cross-encoder relevance as the query--document distance and embedding distance between documents.

	\paragraph{Manipulating and measuring redundancy.}
	 Injection adds to each pool a fraction $\rho$ of sentence-shuffled near-duplicate copies that retain their source's identity (so coverage metrics count them as redundant, not as new evidence) drawn half from gold passages and half uniformly. Sampling half the copies from gold passages is a deliberately adversarial choice: it constructs exactly the mechanism by which duplicates of highly relevant material crowd out complementary evidence, so the injection sweeps are read as controlled stress tests of that mechanism rather than as estimates of the average effect of naturally occurring redundancy. The chunk-overlap sweeps below provide the naturalistic counterpart. The perturbation permutes the sentence units of a punctuation-based splitter without deleting or editing text, so each copy preserves its source verbatim up to unit order (a boundary inside an abbreviation is the one case that separates a multi-word span); the identity-based coverage and relevance metrics are unaffected. Pools are re-embedded after injection and the candidate pool size is held fixed throughout: each query's candidate count is frozen at its clean value $\min\{m,\text{clean pool size}\}$ before any manipulation. At every level and seed, the post-truncation candidate count equals this frozen target so that redundancy is never confounded with pool size.\footnote{The originally executed sweeps truncated to $\min\{m,\text{transformed pool size}\}$, which can enlarge the pool when a query's clean pool is smaller than $m$. All thirteen affected sweeps (twelve run configurations, SciFact holding both an injection and a chunking sweep) were re-executed under the fixed-size policy, and every query, level and seed passed the assertion. The sweep-level method effects and interactions of Section~\ref{sec:rq2} are computed from these fixed-pool re-executions; transfer, oracle, regret and generation inference uses the originally executed runs (Appendix~\ref{app:stats}).}

	 Because copies share an identity, NDCG credits only the first retrieved copy and hyperparameters are tuned on S-Recall@$k$ (the $\alpha=0.5$ gain of $\alpha$-NDCG@$k$ would otherwise reward repeated copies). Every injection and chunk-overlap sweep uses the grids $\lambda\in\{0.3,0.5,0.7,0.9\}$ and $\gamma\in\{-0.3,\allowbreak -0.2,\allowbreak -0.1,\allowbreak -0.05,\allowbreak 0,\allowbreak 0.05,\allowbreak 0.1,\allowbreak 0.2,\allowbreak 0.3,\allowbreak 0.5,\allowbreak 1\}$, which refine the default grids of Section~\ref{sec:baselines} and extend their upper ends ($\lambda=0.9$, margins up to $1$) after validation-selected operating points of early HotpotQA sweeps sat on a boundary. All tuning is re-derived over these fixed grids from the frozen validation splits, so no reported statistic depends on when the extension was decided. The headline sweep refines the grid where the regime turns, $\rho\in\{0,0.025,0.05,0.1,0.15,0.25,0.5,1\}$, on HotpotQA with \texttt{bge-m3} (all $7{,}405$ queries), replicated with \texttt{Qwen3-Embedding-4B} and \texttt{all-MiniLM-L6-v2} (encoder-capacity axis) and on SciFact and FiQA-2018 ($k=10$), the single- and multi-evidence ends of the BEIR tasks. Overlapping chunking instead re-chunks each pool passage into $60$-word sliding windows at overlap $\in\{0,0.25,0.5,0.75\}$, inheriting source identity with no text perturbed, the mechanism of production ingestion pipelines. The same frozen per-query pool size target applies. A selected window inherits its source's gold identity even when the $60$-word window omits the answer-bearing span, an optimistic convention for coverage metrics. Both protocols report the measured redundancy \eqref{eq:pool-redundancy}, placing synthetic and natural sweeps on one axis.

	\paragraph{Rule and oracle.}
	 The decision rule \eqref{eq:rule} fixes $\tau=\min\{h,k\}$, where $h$ is the number of evidence pieces a given query requires, and selects the fallback $D^{*}$ by maximizing mean S-Recall@$5$ on validation queries pooled across all injection levels (never observing the level identity). As a check on the anchoring, the threshold is additionally tuned freely over $[1,k]$ jointly with $D$. The evidence requirement is $h=2$ on the two-hop datasets, the per-query hop count on MuSiQue, $h=1$ on the single-answer datasets and the query's relevant-set size on the graded-relevance BEIR tasks. On TREC-COVID the uncapped $h$ exceeds $k$ for essentially every query, so there the budgeted rule stays close to always-diversify and does not exercise the diagnostic. Because $h$ is read from gold annotations, all gold-$h$ rows are oracle analyses in the sense of Section~\ref{sec:rule}. The frozen rule is then evaluated against always-$k$-NN, always-$D^{*}$, the per-level tuned best method and the all-methods oracle. The per-level tuning selects, on each injection level's validation split, the best method and operating point for that level. Both oracles operate per query with ground-truth access: the per-method oracle \citep{khan2026dfrag} sweeps one method's grid and records the best value achievable, and the all-methods oracle additionally chooses the method, upper-bounding any test-time policy that routes among the evaluated methods and operating points. A new selector, a changed candidate pool, query decomposition or a jointly trained retriever is not bounded by it. A query counts as improved when the oracle's S-Recall strictly exceeds $k$-NN's (library in Appendix~\ref{app:stats}).

	\paragraph{Generation.}
	 Generation experiments attach the reader to the sweep, generating an answer from each method's top-$k$ at every level and scoring it with the generation metrics of Section~\ref{sec:metrics-eval}. The decision rule's answers are reconstructed per query without extra generation, returning the diversifier's answer when the trigger fires and the $k$-NN answer otherwise.

	\section{Results}
	\label{sec:results}

	\subsection{RQ1: Diversification on standard benchmarks}
	\label{sec:rq1}

	Table~\ref{tab:rq1-hotpot} reports controlled reranking on HotpotQA with \texttt{bge-m3} (test splits, three seeds, paired inference against $k$-NN). Every classical diversifier trades relevance for spread, losing between $0.7$ and $33.0$ S-Recall@5 points, significantly for each method (even the mildest loss, \MMR's at $\lambda=0.9$, carries the $95\%$ CI $[-1.05,-0.54]$). Crucially, the intent-aware metrics fall with relevance as APD and Vendi rise. The tuned \RNGSCORE sits on top of $k$-NN by $+0.2$ S-Recall points ($95\%$ CI $[0.04,0.34]$), not significant after Holm adjustment and inside the materiality tolerance either way. On these pools, the correct amount of diversification is approximately none.

	\begin{table}[ht]
		\centering
		\small
		\caption{Controlled reranking on HotpotQA (MDR top-100 pools, \texttt{bge-m3}, $m=100$, $k=5$; clean level of the sweep of Section~\ref{sec:rq2}, on its grids). The four left metrics are objectives (best in bold, second best underlined); APD and Vendi are descriptive spread statistics, not objectives; EM and F1: one deterministic \texttt{Qwen3.8-27B} run, seed-$0$ test split ($5{,}924$ queries).}
		\label{tab:rq1-hotpot}
		\setlength{\tabcolsep}{4pt}
		\begin{tabular}{lcccc|cc|cc}
			\hline
			\textbf{Method} & \textbf{Recall@5} & \textbf{NDCG@5} & \textbf{$\alpha$-NDCG@5} & \textbf{S-Recall@5} & \textbf{APD} & \textbf{Vendi} & \textbf{EM} & \textbf{F1} \\
			\hline
			$k$-NN          & \underline{74.8} & \underline{73.9} & \underline{76.5} & \underline{78.7} & 0.450 & 3.08 & \underline{50.8} & \textbf{63.6} \\
			\MMR ($\lambda=0.3$) & 45.3 & 53.3 & 55.4 & 48.1 & 0.630 & 4.00 & 36.4 & 47.8 \\
			\MMR ($\lambda=0.5$) & 55.6 & 60.3 & 62.5 & 58.6 & 0.574 & 3.73 & 41.6 & 53.5 \\
			\MMR ($\lambda=0.7$) & 69.2 & 70.0 & 72.5 & 72.8 & 0.505 & 3.38 & 47.6 & 59.9 \\
			\MMR ($\lambda=0.9$) & 74.1 & 73.5 & 76.1 & 78.0 & 0.463 & 3.16 & 50.3 & 62.9 \\
			Maxmin       & 42.9 & 51.7 & 53.8 & 45.7 & 0.625 & 3.85 & 36.0 & 47.4 \\
			Greedy \DPP  & 65.6 & 68.0 & 70.5 & 69.1 & 0.522 & 3.47 & 45.9 & 58.0 \\
			\RNGSCORE ($\gamma^{*}=0.2$)   & \textbf{75.1} & \textbf{74.2} & \textbf{76.8} & \textbf{79.0} & 0.462 & 3.14 & \textbf{50.9} & \textbf{63.6} \\
			\hline
		\end{tabular}
	\end{table}

	The retrieval losses propagate to generated answers (EM/F1 columns of Table~\ref{tab:rq1-hotpot}). On the clean HotpotQA pools, every fixed-pressure diversifier loses EM against $k$-NN, while the tuned \RNGSCORE sits $0.1$ EM points above $k$-NN. Section~\ref{sec:rq5} carries this into the redundant regime, where the picture inverts.

	The single-stage behavior of 2WikiMultiHopQA, MuSiQue and the NQ-Open single-answer control tracks their measured pool redundancy (Appendix~\ref{app:single-stage}): 2WikiMultiHopQA, whose natural pools measure below the crossover of Section~\ref{sec:rq2}, reproduces these conclusions, while MuSiQue's distractor pools measure at the crossover fraction and mild diversification pays about two S-Recall points. Under the cross-encoder pipelines below, all four datasets lead to the same conclusion, i.e. diversification hurts and top-$k$ selection wins.

	Table~\ref{tab:rq1-beir} extends the study to five BEIR tasks. The conclusion is uniform: every fixed-pressure diversifier loses on both NDCG@10 and $\alpha$-NDCG@10 metrics. The tuned \RNGSCORE is the exception on the three seed-replicated tasks, where validation drives its margin into or near the fallback region and leaves it within one point of $k$-NN.

	\begin{table}[ht]
		\centering
		\small
		\caption{NDCG@10 / $\alpha$-NDCG@10 on five BEIR tasks (\texttt{bge-m3}, $k=10$; \MMR\ at $\lambda=0.5$; best in bold, second best underlined). SciFact, FiQA and TREC-COVID are seed-replicated; conclusions on these small query sets are exploratory (Appendix~\ref{app:stats}). ArguAna and Touch\'e are single-run.}
		\label{tab:rq1-beir}
		\begin{tabular}{lccccc}
			\hline
			\textbf{Method} & \textbf{SciFact} & \textbf{FiQA} & \textbf{TREC-COVID} & \textbf{ArguAna} & \textbf{Touch\'e} \\
			\hline
			$k$-NN          & \textbf{64.1}/\textbf{64.1} & \textbf{41.7}/\underline{46.7} & \textbf{66.1}/\textbf{70.2} & \textbf{38.3}/\textbf{38.3} & \textbf{25.2}/\textbf{27.8} \\
			\MMR         & 52.9/52.9 & 29.6/33.0 & 37.0/40.3 & 0.2/0.2 & 10.9/12.9 \\
			Maxmin       & 49.4/49.4 & 27.2/30.2 & 46.7/52.0 & 0.9/0.9 & 9.2/11.7 \\
			Greedy \DPP  & 56.0/56.0 & 30.9/34.4 & 39.8/43.8 & 1.0/1.0 & 11.0/13.2 \\
			\RNGSCORE    & \underline{63.4}/\underline{63.5} & \textbf{41.7}/\textbf{46.8} & \underline{65.6}/\underline{69.6} & \underline{26.4}/\underline{26.4} & \underline{17.8}/\underline{19.6} \\
			\hline
		\end{tabular}
	\end{table}

	\paragraph{After a strong cross-encoder.}
	\label{sec:rq1-crossencoder}
	Production pipelines rarely diversify a raw dense pool, they usually rerank first. We therefore repeat the test after the cross-encoder pipeline of Section~\ref{sec:protocol}. Table~\ref{tab:rq2-ce} reports the result on four datasets.

	\begin{table}[ht]
		\centering
		\small
		\caption{Cross-encoder pipelines, $k=5$ (\MMR\ at its best $\lambda=0.7$; CE+\RNGSCORE\ is the tuned dissimilarity variant of Appendix~\ref{app:ce-variants}; the non-HotpotQA datasets are single-run). Best per dataset and column in bold; a tie occupies both ranks, so cells below a tied best are not underlined.}
		\label{tab:rq2-ce}
		\begin{tabular}{llcccc}
			\hline
			\textbf{Dataset} & \textbf{Method (post-CE)} & \textbf{Recall@5} & \textbf{NDCG@5} & \textbf{EM} & \textbf{F1} \\
			\hline
			\multirow{4}{*}{HotpotQA}
			& CE top-$k$        & \textbf{89.6} & \textbf{87.9} & \textbf{41.0} & \textbf{53.3} \\
			& CE+\MMR           & 84.6 & 83.4 & 38.2 & 50.0 \\
			& CE+\DPP           & 80.0 & 80.8 & 36.2 & 47.7 \\
			& CE+\RNGSCORE      & \textbf{89.6} & \textbf{87.9} & \textbf{41.0} & \textbf{53.3} \\
			\hline
			\multirow{4}{*}{2WikiMultiHopQA}
			& CE top-$k$        & \textbf{90.6} & \textbf{90.3} & \textbf{31.3} & \textbf{37.9} \\
			& CE+\MMR           & 71.0 & 76.3 & 25.0 & 29.6 \\
			& CE+\DPP           & 71.2 & 76.5 & 25.2 & 29.9 \\
			& CE+\RNGSCORE      & \textbf{90.6} & \textbf{90.3} & \textbf{31.3} & \textbf{37.9} \\
			\hline
			\multirow{4}{*}{MuSiQue}
			& CE top-$k$        & \underline{74.3} & \textbf{77.6} & \underline{26.0} & \underline{34.1} \\
			& CE+\MMR           & 52.7 & 61.0 & 18.5 & 25.6 \\
			& CE+\DPP           & 53.1 & 61.6 & 19.8 & 26.7 \\
			& CE+\RNGSCORE      & \textbf{74.4} & \textbf{77.6} & \textbf{26.3} & \textbf{34.3} \\
			\hline
			\multirow{4}{*}{NQ-Open}
			& CE top-$k$        & \textbf{23.5} & \textbf{65.8} & \textbf{33.0} & \textbf{43.4} \\
			& CE+\MMR           & 16.3 & 51.9 & 29.7 & 40.0 \\
			& CE+\DPP           & 13.3 & 45.3 & 29.8 & 39.7 \\
			& CE+\RNGSCORE      & \textbf{23.5} & \textbf{65.8} & \textbf{33.0} & \textbf{43.4} \\
			\hline
		\end{tabular}
	\end{table}

	The cross-encoder only strengthens the negative result: CE top-$k$ wins every comparison with the fixed-pressure diversifiers. The tuned \RNGSCORE tracks the cross-encoder ranking within a fraction of a point on every dataset, sitting at the CE top-$k$ values and nominally above them on MuSiQue: validation lands its margin at or near one of the two fallback limits throughout (per-dataset margins in Appendix~\ref{app:ce-variants}). In every experiment of this subsection the annotation-free and intent-aware diversity metrics dissociate, diversifiers raising APD and Vendi while lowering $\alpha$-NDCG and S-Recall.

	\subsection{RQ2: Manipulating pool redundancy}
	\label{sec:rq2}

	\paragraph{Misspecification downside.}
	\label{sec:rq2-regret}
	Table~\ref{tab:rq3-regret} reports, for each method, the best and worst S-Recall@5 over its hyperparameter grid on clean ($\rho=0$) and heavily redundant ($\rho=1$) pools and the resulting misspecification downside \eqref{eq:regret}.

	\begin{table}[ht]
		\centering
		\small
		\caption{Misspecification downside \eqref{eq:regret} and common-oracle minimax regret \eqref{eq:minimax} on HotpotQA, with the best and worst S-Recall@5 over each method's grid on clean and heavily redundant pools (the $k$-NN row of the downside column is constructed against the best tuned method per regime and is not directly comparable to the others). Lower is better in the last two columns. Query-bootstrap $95\%$ intervals (computed, like the regret column, from query-level test means, hence a few tenths off the displayed seed-mean cells): downside \MMR\ $[30.2,31.5]$, Maxmin $[32.6,34.0]$, \DPP\ $[9.2,10.3]$, \RNGSCORE\ $[3.5,4.4]$; regret $k$-NN $[21.6,22.8]$, \MMR\ $[5.9,6.8]$, Maxmin $[32.9,34.3]$, \DPP\ $[9.6,10.6]$, \RNGSCORE\ $[8.9,10.0]$. Best per column in bold, second best underlined.}
		\label{tab:rq3-regret}
		\begin{tabular}{lcc|cc|cc}
			\hline
			& \multicolumn{2}{c|}{$\rho=0$ (clean)} & \multicolumn{2}{c|}{$\rho=1$ (redundant)} & & \\
			\textbf{Method} & \textbf{best} & \textbf{worst} & \textbf{best} & \textbf{worst} & \textbf{max.\ downside} & \textbf{minimax regret} \\
			\hline
			$k$-NN          & \multicolumn{2}{c|}{\underline{78.7}} & \multicolumn{2}{c|}{52.6} & 22.0 & 22.2 \\
			\MMR         & 78.0 & 48.1 & \textbf{74.6} & 53.5 & 30.6 & \textbf{6.3} \\
			Maxmin       & 45.7 & 45.7 & 50.6 & 50.6 & 33.0 & 33.6 \\
			Greedy \DPP  & 69.1 & 69.1 & \underline{72.1} & 72.1 & \underline{9.6} & 10.0 \\
			\RNGSCORE    & \textbf{79.0} & 74.8 & 65.2 & 52.8 & \textbf{3.9} & \underline{9.5} \\
			\hline
		\end{tabular}
	\end{table}

	Three facts stand out. First, every fixed commitment carries regime risk, including the default: always-$k$-NN gives up $22.0$ points in the redundant regime, the mirror image of \MMR's worst clean-pool toll ($30.6$ points). Second, the harm range of the fixed-pressure methods is set entirely by their operating point: against the $k$-NN rows of Tables~\ref{tab:rq1-hotpot} and~\ref{tab:rq3-regret}, the $\lambda=0.7$ that is optimal at $\rho=1$ costs $5.9$ points at $\rho=0$ while the $\lambda=0.3$ value that enforces a more diverse retrieval costs $30.6$. Third, the relative-neighborhood family is the only method whose worst grid point never falls more than $3.9$ points below the baseline in either regime. The two summary columns answer different questions and order the methods differently: anchored at the $k$-NN default, the downside favors the \RNGSCORE, whose fallback keeps every operating point close to the baseline; anchored at the per-regime oracle, the minimax regret favors the tuned \MMR\ ($6.3$ points, $95\%$ CI $[5.9,6.8]$, at $\lambda=0.7$), which comes closest to the optimum past the crossover, with the \RNGSCORE\ second ($9.5$ points $[8.9,10.0]$). Both remain grid-dependent (intervals in the caption of Table~\ref{tab:rq3-regret}).

	\paragraph{Injected redundancy.}

	Table~\ref{tab:rq4-redundancy} reports the sweep on HotpotQA with \texttt{bge-m3} (three seeds, all $7{,}405$ queries executed, test splits reported, operating points tuned per $\rho$). At $\rho=0$ the conclusions of Section~\ref{sec:rq1} are reproduced. Injection then reverses the ordering. The paired effect of the tuned \MMR against $k$-NN changes sign between the measured grid levels $\rho=0.05$ and $\rho=0.1$: $-0.02$ points ($95\%$ CI $[-0.26,0.23]$) at $\rho=0.05$ and $+1.83$ points ($[1.37,2.31]$) at $\rho=0.1$, bracketing the crossover between measured near-duplicate fractions $0.0012$ and $0.0030$. Past the crossover, $k$-NN S-Recall falls steeply as the Vendi score of its selection collapses from $3.1$ toward $1.2$, both tuned diversifiers improve on $k$-NN by wide and significant margins (at $\rho=1$, $+22.2$ points $[21.6,22.8]$ for \MMR$^{*}$ and $+12.7$ points $[12.2,13.2]$ for the \RNGSCORE) and parameter-free \DPP joins them from $\rho=0.15$ (query-agnostic Maxmin remains below the baseline throughout).

	\begin{table}[ht]
		\centering
		\small
		\caption{Redundancy-injection sweep, HotpotQA (\texttt{bge-m3}, $k=5$); S-Recall@5, operating points tuned per $\rho$. Second header row: measured pool redundancy \eqref{eq:pool-redundancy}. Best per column in bold, second best underlined.}
		\label{tab:rq4-redundancy}
		\setlength{\tabcolsep}{3pt}
		\begin{tabular}{lcccccccc}
			\hline
			\textbf{Method} & $\rho{=}0$ & $0.025$ & $0.05$ & $0.1$ & $0.15$ & $0.25$ & $0.5$ & $1.0$ \\
			& \scriptsize 0 & \scriptsize.0006 & \scriptsize.0012 & \scriptsize.0030 & \scriptsize.0054 & \scriptsize.012 & \scriptsize.042 & \scriptsize.147 \\
			\hline
			$k$-NN          & \underline{79.0} & \underline{77.3} & \underline{76.9} & 71.4 & 62.1 & 54.5 & 53.0 & 52.7 \\
			\MMR ($\lambda^{*}$)  & 78.2 & 77.0 & 76.9 & \underline{73.2} & \underline{73.3} & \textbf{73.5} & \textbf{73.9} & \textbf{74.9} \\
			Maxmin       & 45.7 & 46.1 & 46.2 & 46.7 & 47.1 & 47.5 & 48.6 & 51.8 \\
			Greedy \DPP  & 69.2 & 69.5 & 69.9 & 70.3 & 70.6 & 70.8 & \underline{71.3} & \underline{72.4} \\
			\RNGSCORE ($\gamma^{*}$)   & \textbf{79.2} & \textbf{77.8} & \textbf{77.8} & \textbf{76.5} & \textbf{74.7} & \underline{71.7} & 68.1 & 65.4 \\
			\hline
		\end{tabular}
	\end{table}

	The sweep further shows that no fixed operating point dominates the whole axis: the \RNGSCORE is the best method through $\rho=0.15$, while \MMR\ overtakes it from $\rho=0.25$ onward yet gives up S-Recall on every clean level.

	The primary test confirms that the effect of diversification itself changes with redundancy, rather than merely differing in significance across levels: on the headline sweep every method's clean-to-heavy improvement is significant, from $+12.5$ points for the \RNGSCORE to $+32.4$ for Maxmin, and the same holds on every other sweep with a single exception (Table~\ref{tab:interactions} in Appendix~\ref{app:stats}).

	Figure~\ref{fig:crossover} replicates the sweep under three encoders on the common measured-redundancy axis. The regime shape is fully reproduced under all three, and the mapping from $\rho$ to measured redundancy agrees within $6\%$ across them, which licenses the common axis. The strong encoders share the sign-change bracket ($\rho=0.05$ to $0.1$), while under the weak encoder the seed-mean effect of the tuned \MMR\ changes sign two grid steps earlier, between the clean level and $\rho=0.025$ ($-0.2$ points at a measured fraction of $0.0001$, $+0.2$ at $0.0006$). The replicas carry no per-level intervals: the earlier sign change is descriptive, the difference in location is not tested, and the encoder and mechanism replications are read as qualitative reproductions of the regime shape with no claim of effect homogeneity.

	\begin{figure}[ht]
		\centering
		\begin{subfigure}[b]{0.48\linewidth}
			\centering
			\includegraphics[width=\linewidth, alt={Line plot of S-Recall against measured pool redundancy on a logarithmic axis for k-NN, MMR*, RNG-Score* and greedy DPP with the bge-m3 encoder. k-NN, MMR* and RNG-Score* start near 79 on clean pools; k-NN collapses to about 53 once redundancy exceeds about 0.003, MMR* stays above 73, RNG-Score* declines at the highest levels, and greedy DPP rises slowly from about 69 to 72.}]{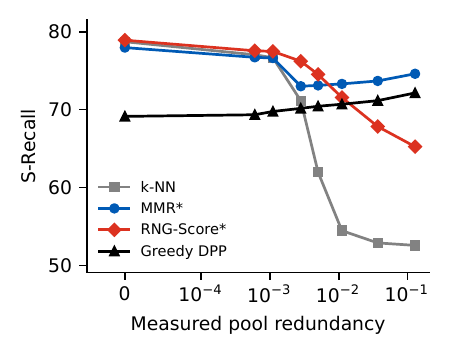}
			\subcaption{\texttt{bge-m3}}
			\label{fig:crossover-bge}
		\end{subfigure}\hfill
		\begin{subfigure}[b]{0.48\linewidth}
			\centering
			\includegraphics[width=\linewidth, alt={Line plot of S-Recall against measured pool redundancy on a logarithmic axis for k-NN, MMR*, RNG-Score* and greedy DPP with the Qwen3-Embedding-4B encoder. k-NN, MMR* and RNG-Score* start near 80 on clean pools; k-NN collapses to about 53 at high redundancy, MMR* and greedy DPP settle near 75 to 77, and RNG-Score* declines to about 67.}]{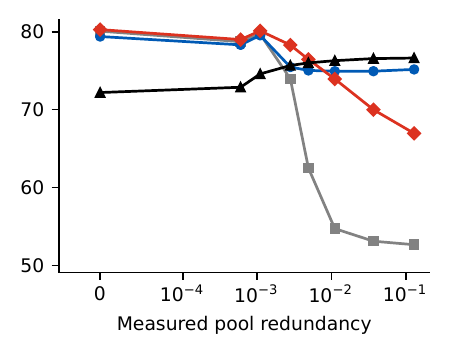}
			\subcaption{\texttt{Qwen3-Embedding-4B}}
			\label{fig:crossover-qwen}
		\end{subfigure}\\[2pt]
		\begin{subfigure}[b]{0.48\linewidth}
			\centering
			\includegraphics[width=\linewidth, alt={Line plot of S-Recall against measured pool redundancy on a logarithmic axis for k-NN, MMR*, RNG-Score* and greedy DPP with the all-MiniLM-L6-v2 encoder. All methods start between 60 and 65; k-NN falls to about 51 at high redundancy, RNG-Score* to about 58, while MMR* and greedy DPP rise to about 66 to 67.}]{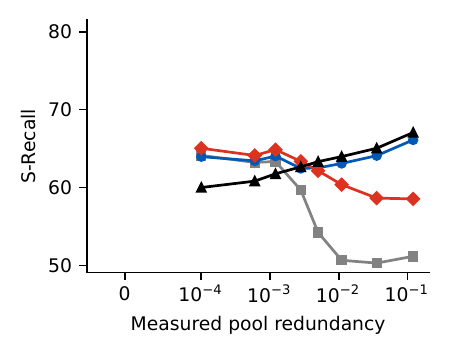}
			\subcaption{\texttt{all-MiniLM-L6-v2}}
			\label{fig:crossover-minilm}
		\end{subfigure}
		\caption{The HotpotQA injection sweep under three encoders: S-Recall
			vs measured near-duplicate pair fraction \eqref{eq:pool-redundancy}, on
			common axes (log-scaled with the clean level kept on-axis).}
		\label{fig:crossover}
	\end{figure}

	\paragraph{Evidence multiplicity on BEIR.}
	Multi-evidence FiQA-2018 replicates the injection experiment on a natural non-QA benchmark (Figure~\ref{fig:crossover-fiqa}). The regime turns exactly as on multi-hop QA: both tuned diversifiers overtake $k$-NN significantly and stably from $\rho=0.1$ onward, by up to $+11.3$ S-Recall@10 points at $\rho=1$, with the tuned margin of the \RNGSCORE at an interior negative value across the sweep (Table~\ref{tab:rq2-margins}). At $\rho=0$ neither differs significantly from $k$-NN. Single-hop SciFact shows the opposite: under the same encoder, neither tuned diversifier differs significantly from $k$-NN at any injection level (Figure~\ref{fig:crossover-scifact}).

	\begin{figure}[ht]
		\centering
		\begin{subfigure}[b]{0.48\linewidth}
			\centering
			\includegraphics[width=\linewidth, alt={Line plot of S-Recall at 10 against measured pool redundancy on a logarithmic axis for k-NN, MMR*, RNG-Score* and greedy DPP on SciFact. k-NN, MMR* and RNG-Score* stay flat between about 80 and 84 across the sweep, while greedy DPP rises from about 69 to about 82 at the heaviest level.}]{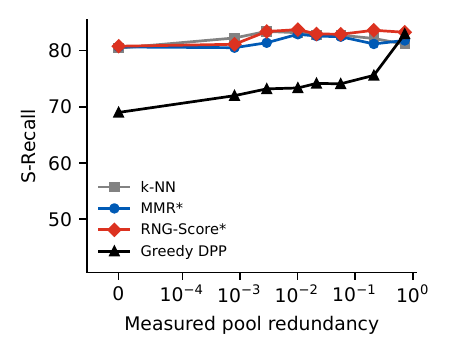}
			\subcaption{SciFact (single-hop)}
			\label{fig:crossover-scifact}
		\end{subfigure}\hfill
		\begin{subfigure}[b]{0.48\linewidth}
			\centering
			\includegraphics[width=\linewidth, alt={Line plot of S-Recall at 10 against measured pool redundancy on a logarithmic axis for k-NN, MMR*, RNG-Score* and greedy DPP on FiQA-2018. k-NN declines from about 57 to 50 as redundancy grows, MMR* and RNG-Score* rise slightly to about 61, and greedy DPP rises from about 42 to 63.}]{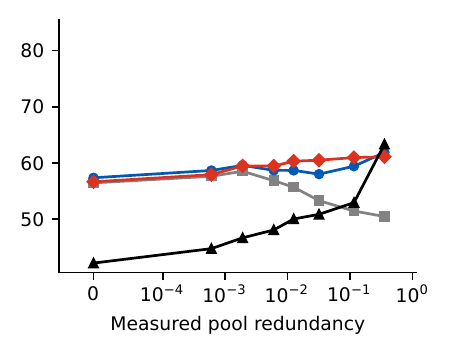}
			\subcaption{FiQA-2018 (multi-evidence)}
			\label{fig:crossover-fiqa}
		\end{subfigure}
		\caption{The BEIR injection sweeps (\texttt{bge-m3}, $k=10$, three seeds):
			S-Recall@10 vs measured near-duplicate pair fraction
			\eqref{eq:pool-redundancy}, on common axes. Single-evidence SciFact
			(\subref{fig:crossover-scifact}) stays flat across the sweep;
			multi-evidence FiQA-2018 (\subref{fig:crossover-fiqa}) reproduces the
			multi-hop collapse-and-crossover pattern.}
		\label{fig:crossover-beir}
	\end{figure}

	\paragraph{The tuned margin as an offline readout.}
	Table~\ref{tab:rq2-margins} collects the validation-selected \RNGSCORE margins across the five injection sweeps. On the headline \texttt{bge-m3} sweep, the margin flips sign at the first injection level, $+0.2$ on the clean pool and negative in every seed thereafter. Under the other encoders, the flip arrives later, with seeds disagreeing in sign exactly in the transition zones. On SciFact, seeds disagree at nearly every level. On FiQA-2018 the margin sits at the interior value $-0.2$ at every level where the seeds agree, including the clean pool: FiQA's natural pools are crowded at similarities well below the near-duplicate cutoff, the $k$-NN selection retaining an effective $3.5$ of $10$ distinct documents (Vendi) while the pair fraction \eqref{eq:pool-redundancy} reads zero.

	\begin{table}[ht]
		\centering
		\small
		\caption{Validation-selected \RNGSCORE margins $\gamma^{*}$ per injection level: modal value across the three seeds where the seeds agree in sign, $\pm$ where they do not.}
		\label{tab:rq2-margins}
		\setlength{\tabcolsep}{4pt}
		\begin{tabular}{lcccccccc}
			\hline
			\textbf{Sweep} & $\rho{=}0$ & $0.025$ & $0.05$ & $0.1$ & $0.15$ & $0.25$ & $0.5$ & $1.0$ \\
			\hline
			HotpotQA, \texttt{bge-m3} & $+0.2$ & $-0.3$ & $-0.3$ & $-0.2$ & $-0.2$ & $-0.2$ & $-0.2$ & $-0.2$ \\
			HotpotQA, \texttt{Qwen3-Embedding-4B} & $+0.3$ & $\pm$ & $\pm$ & $-0.2$ & $-0.1$ & $-0.1$ & $-0.1$ & $-0.1$ \\
			HotpotQA, \texttt{all-MiniLM-L6-v2} & $+0.3$ & $+0.3$ & $+0.2$ & $+0.2$ & $+0.2$ & $\pm$ & $-0.1$ & $-0.1$ \\
			SciFact, \texttt{bge-m3} & $-0.2$ & $\pm$ & $\pm$ & $-0.3$ & $\pm$ & $\pm$ & $\pm$ & $\pm$ \\
			FiQA-2018, \texttt{bge-m3} & $-0.2$ & $\pm$ & $-0.2$ & $-0.2$ & $-0.2$ & $-0.2$ & $-0.2$ & $-0.2$ \\
			\hline
		\end{tabular}
	\end{table}

	\paragraph{Natural redundancy through chunk overlap.}
	Table~\ref{tab:rq4-chunking} reports the chunk-overlap sweep of Section~\ref{sec:protocol}, where redundancy arises naturally from overlapping sliding windows: S-Recall is reported for every method both on the non-overlapping pool and at the heaviest overlap. Across the six sweeps $k$-NN drifts down by up to $4.0$ points at the heaviest overlap while the relevance-aware diversifiers hold essentially flat or improve, most steeply on the multi-hop 2WikiMultiHopQA and MuSiQue pools, where greedy \DPP\ gains $7.2$ and $9.9$ points. Moreover, with the exception of \DPP on HotpotQA, the relevance-aware diversifiers all score higher than $k$-NN in the more redundant setting. The primary test is significant for every method on every chunk-overlap sweep (Table~\ref{tab:interactions}): the natural mechanism reproduces the injected one at a smaller amplitude.

	\begin{table}[ht]
		\centering
		\small
		\caption{Chunk-overlap sweep (\texttt{bge-m3}, $k=5$; $k=10$ for SciFact and TREC-COVID; 2Wiki abbreviates 2WikiMultiHopQA). Each cell reports the row method's S-Recall@$k$ on the non-overlapping pool~/~at overlap $0.75$ (tuned operating points per level); the header rows give each dataset's hop count and measured redundancy \eqref{eq:pool-redundancy} at overlap $0.75$. Best per dataset in bold, second best underlined (ties both bold and occupying both ranks). The NQ-Open sweep is single-run; conclusions on the small SciFact and TREC-COVID query sets are exploratory (Appendix~\ref{app:stats}); saturated TREC-COVID (mean $493$ relevant documents per query) has no hop count.}
		\label{tab:rq4-chunking}
		\setlength{\tabcolsep}{3.5pt}
		\begin{tabular}{lcccccc}
			\hline
			\textbf{Method} & \textbf{HotpotQA} & \textbf{2Wiki} & \textbf{MuSiQue} & \textbf{SciFact} & \textbf{NQ-Open} & \textbf{TREC-COVID} \\
			{\scriptsize hops} & \scriptsize 2 & \scriptsize 2 & \scriptsize 2--4 & \scriptsize 1 & \scriptsize 1 & \scriptsize sat. \\
			{\scriptsize red$_{.75}$} & \scriptsize .001 & \scriptsize .011 & \scriptsize .015 & \scriptsize .002 & \scriptsize .005 & \scriptsize .006 \\
			\hline
			$k$-NN       & \underline{78.1}/75.0 & \underline{88.5}/87.4 & 74.1/73.2 & \textbf{77.2}/73.2 & 92.1/89.7 & \textbf{21.8}/23.1 \\
			Maxmin       & 46.5/48.3 & 74.7/83.6 & 56.9/69.5 & 57.8/62.2 & 86.8/87.6 & 14.9/21.9 \\
			Greedy \DPP  & 69.5/70.9 & 87.2/\textbf{94.4} & 76.2/\textbf{86.1} & 73.3/74.1 & 91.6/92.0 & 18.6/25.9 \\
			\MMR$^{*}$   & 77.6/\underline{76.6} & 88.2/91.5 & \underline{76.3}/85.0 & \underline{77.0}/\underline{76.1} & \underline{92.8}/\textbf{93.0} & \underline{21.4}/\underline{27.2} \\
			\RNGSCORE$^{*}$ & \textbf{78.2}/\textbf{78.2} & \textbf{88.7}/\underline{93.7} & \textbf{76.7}/\underline{85.7} & 75.9/\textbf{76.4} & \textbf{92.9}/\textbf{93.0} & 21.3/\textbf{27.7} \\
			\hline
		\end{tabular}
	\end{table}

	\paragraph{The saturated-pool regime.}
	A third point on the evidence-multiplicity axis is the \emph{saturated} pool, where almost every candidate is relevant. TREC-COVID is the extreme case with a mean of 493 relevant documents per query. Under the chunk-overlap sweep (\texttt{bge-m3}, $k=10$), $k$-NN coverage does not collapse, S-Recall@10 staying flat to slightly rising: there is always another relevant document to take a crowded slot. Yet, at the highest overlap, the tuned \RNGSCORE significantly leads $k$-NN by $4.6$ points (last row of Table~\ref{tab:rq4-chunking}). Saturation thus protects relevance-only selection from redundancy without eliminating the diversification gain, and the frozen rule of \ref{rq:rule} shows no significant difference from $k$-NN on the natural pools while gaining significantly at the heaviest overlap (Table~\ref{tab:rq6-transfer}).

	All manipulations thus give the same reading of \ref{rq:redundancy}: raising pool redundancy does create a regime in which diversification pays, at a crossover that the deployment-level statistic \eqref{eq:pool-redundancy} localizes on the measured axis, and the risk of committing to a method without knowing the regime is the misspecification downside of Table~\ref{tab:rq3-regret}, up to $33$ coverage points for the popular fixed commitments.

	\subsection{RQ3: Oracle headroom}
	\label{sec:rq3}

	Results from the oracle protocol of Section~\ref{sec:protocol} are reported in Figure~\ref{fig:oracle}. On clean pools of HotpotQA, the headroom is negligible: an oracle with ground-truth access that chooses freely among all methods and grid values per-query improves S-Recall over $k$-NN on only $5.8\%$ of queries for a gain of $3.1$ points. The headroom starts to open at the lower edge of the crossover bracket (a measured redundancy of $0.0012$, reached at $\rho=0.05$), where the oracle beats $k$-NN on $9.2\%$ of queries, rises to $19.1\%$ at the next level and peaks when redundancy is maximal, where it improves on $k$-NN on over half of queries (its per-level coverage appears as the oracle row of Table~\ref{tab:rq6-rule}). On single-hop SciFact, it stays at $4$--$10\%$ across the whole sweep. On multi-evidence FiQA-2018, whose crowded natural pools already leave the oracle $8.9\%$ of queries to improve on the clean level, the headroom climbs monotonically to $34.8\%$ at $\rho=1$. On saturated TREC-COVID, the injection sweep confirms the chunk-overlap picture: the oracle already beats $k$-NN on a quarter of the queries on the natural pools and on $55$--$99\%$ across the injection levels.

	The chunk-overlap sweeps reproduce the pattern (Figure~\ref{fig:oracle-chunking}): headroom rises with overlap on every multi-evidence dataset, stays low on the single-hop datasets ($5$--$7\%$ on NQ-Open, $7$--$12\%$ on SciFact) and is largest on saturated TREC-COVID ($55$--$92\%$), where a better selection almost always exists even though $k$-NN never collapses.

	\begin{figure}[ht]
		\centering
		\begin{subfigure}[b]{0.49\linewidth}
			\centering
			\includegraphics[width=\linewidth, alt={Line plot of the percentage of test queries on which the all-methods oracle beats k-NN against the injected fraction rho for HotpotQA, 2Wiki, MuSiQue, SciFact, FiQA-2018 and TREC-COVID. The percentage increases with rho on every dataset, from about 25 to 99 percent on TREC-COVID, about 28 to 77 on MuSiQue, about 6 to 56 on HotpotQA and 2Wiki, and stays near 5 to 10 percent on SciFact.}]{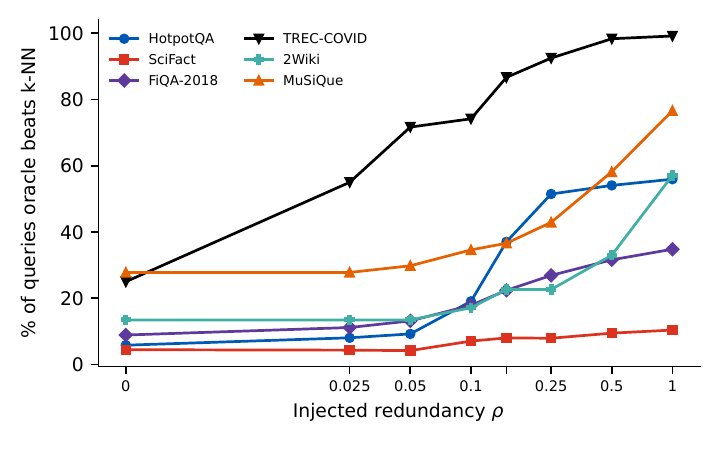}
			\subcaption{Injection sweeps}
			\label{fig:oracle-injection}
		\end{subfigure}\hfill
		\begin{subfigure}[b]{0.49\linewidth}
			\centering
			\includegraphics[width=\linewidth, alt={Line plot of the percentage of test queries on which the all-methods oracle beats k-NN against chunk overlap from 0 to 0.75 for the same six datasets. The percentage increases with overlap on every dataset, from about 55 to 92 percent on TREC-COVID, about 34 to 52 on MuSiQue, about 20 to 39 on 2Wiki, and remains below about 20 percent on FiQA-2018, HotpotQA and SciFact.}]{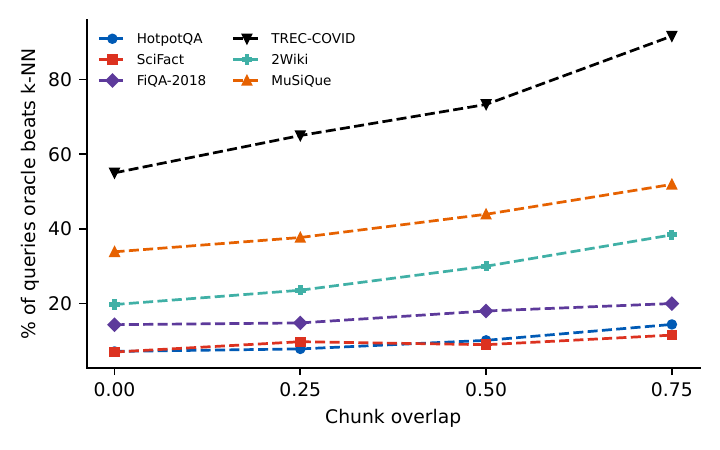}
			\subcaption{Chunk-overlap sweeps}
			\label{fig:oracle-chunking}
		\end{subfigure}
		\caption{Oracle headroom: percentage of test queries on which the
			all-methods oracle beats $k$-NN (HotpotQA, 2Wiki and MuSiQue at
			$k=5$; SciFact, FiQA-2018 and TREC-COVID at $k=10$; three seeds
			throughout).
			(\subref{fig:oracle-injection}) Injection sweeps vs injected fraction
			$\rho$ (log-scaled axis with the clean level kept on-axis);
			(\subref{fig:oracle-chunking}) Chunk-overlap sweeps vs overlap.}
		\label{fig:oracle}
	\end{figure}

	\subsection{RQ4: The decision rule}
	\label{sec:rq4}

	Under the tuning protocol of Section~\ref{sec:protocol}, applied once to the \texttt{bge-m3} HotpotQA injection sweep with the fallback diversifier ranging over every method and grid value, free threshold tuning selects $\tau^{*}=1.83$ ($1.96/1.71/1.83$ per seed) with $D^{*}=\MMR(\lambda=0.7)$ in every seed. Retuned on the $k=10$ replica of the same sweep, it selects $\tau^{*}=1.9$ ($1.9/1.9/1.8$ per seed). The tuned threshold thus stays at the evidence requirement $h=2$ for both $k=5$ and $k=10$. The selection is insensitive: every threshold in $\tau\in[1.5,2.5]$ performs within half a point of the optimum at $k=5$ (Figure~\ref{fig:threshold-curve}) and replacing the Vendi trigger with the mean pairwise distance of the $k$-NN selection reproduces the headroom capture within half a point at every level. We therefore freeze the deployed rule at $\tau=h$, that is $2$ on the two-hop benchmarks, and use it everywhere below. The threshold is interpretable: we diversify only when the selection carries fewer effectively distinct documents than the query needs.

	\begin{figure}[ht]
		\centering
		\begin{subfigure}[b]{0.46\linewidth}
			\centering
			\includegraphics[width=\linewidth, alt={Curve of pooled validation S-Recall at 5 against the trigger threshold tau from 1 to 4 on HotpotQA. The curve rises steeply from the always k-NN value near 66 to a broad plateau near 76 around tau between 1.8 and 2, then declines slowly toward the always MMR(0.7) value near 74. Vertical dashed lines mark the tuned value tau* = 1.83 and the deployed value tau = h = 2.}]{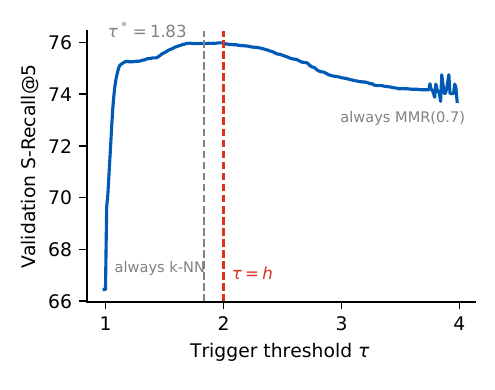}
			\subcaption{Threshold tuning curve}
			\label{fig:threshold-curve}
		\end{subfigure}\hfill
		\begin{subfigure}[b]{0.52\linewidth}
			\centering
			\includegraphics[width=\linewidth, alt={Line plot of test S-Recall at 5 against injected redundancy rho on HotpotQA for always k-NN, always MMR(0.7), the frozen rule with tau equal to h, and the per-query oracle. Always k-NN falls from about 79 to 53, always MMR stays between 73 and 75, the rule matches k-NN on clean pools and MMR on redundant pools, and the oracle stays near 80 to 82.}]{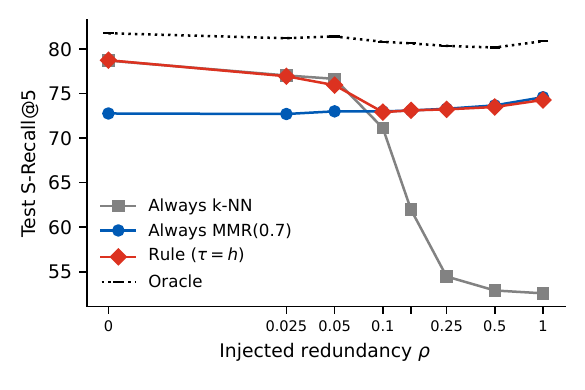}
			\subcaption{Frozen rule across injection levels}
			\label{fig:threshold-rule}
		\end{subfigure}
		\caption{The decision rule on the HotpotQA sweep.
			(\subref{fig:threshold-curve}) Pooled validation S-Recall@5 as the trigger
			threshold $\tau$ varies continuously.
			(\subref{fig:threshold-rule}) The frozen rule per injection level
			(log-scaled $\rho$ axis, ticks at the sweep levels).}
		\label{fig:threshold}
	\end{figure}

	\begin{table}[ht]
		\centering
		\small
		\caption{The frozen rule ($\tau=h=2$, $D^{*}=\MMR(0.7)$) per injection level, S-Recall@5. ``Per-level tuned'' knows the regime; the oracle is the all-methods per-query bound of Section~\ref{sec:rq3}; the trigger rate is the fraction of queries on which the rule diversifies. Among the three policies, best in bold, second best underlined.}
		\label{tab:rq6-rule}
		\setlength{\tabcolsep}{4pt}
		\begin{tabular}{lccccccccc}
			\hline
			\textbf{Selector} & $\rho{=}0$ & $0.025$ & $0.05$ & $0.1$ & $0.15$ & $0.25$ & $0.5$ & $1.0$ & \textbf{pool} \\
			\hline
			always $k$-NN          & \textbf{78.7} & \textbf{77.0} & \textbf{76.7} & 71.1 & 62.0 & 54.5 & 52.9 & 52.6 & 65.7 \\
			always $D^{*}$      & 72.8 & 72.7 & 73.0 & \textbf{73.0} & \textbf{73.1} & \textbf{73.3} & \textbf{73.7} & \textbf{74.6} & \underline{73.3} \\
			rule                & \textbf{78.7} & \underline{76.9} & \underline{76.0} & \underline{72.9} & \textbf{73.1} & \underline{73.2} & \underline{73.5} & \underline{74.3} & \textbf{74.8} \\
			\quad trigger\,\%   & 0.7 & 3.3 & 15.5 & 81.0 & 87.8 & 90.0 & 91.4 & 92.9 & \\
			\hline
			\multicolumn{10}{l}{\emph{Reference bounds (not deployable policies)}}\\
			per-level tuned     & 78.9 & 77.6 & 77.5 & 76.2 & 74.5 & 73.3 & 73.7 & 74.6 & 75.8 \\
			oracle              & 81.8 & 81.2 & 81.4 & 80.8 & 80.6 & 80.4 & 80.2 & 80.9 & 80.9 \\
			\hline
		\end{tabular}
	\end{table}

	Table~\ref{tab:rq6-rule} is the central practical result of the paper. The frozen rule tracks the upper envelope of the two fixed policies across the whole axis. On the three cleanest levels, it fires on under $16\%$ of queries and stays within $0.7$ points of $k$-NN, avoiding the toll the diversifier pays there; from $\rho=0.1$ onward it fires on $81$--$93\%$ and lands within $0.3$ points of always-\MMR. Pooled across levels it captures $60\%$ ($95\%$ CI $[58,62]$) of the all-methods oracle headroom of $15.3$ points $[15.0,15.7]$, a bound that uses ground-truth labels per query, and $91\%$ $[89,92]$ of the gain of the per-level tuned selector.

	\paragraph{Transfer of the frozen rule.}
	We apply the frozen rule with no retuning to fifteen further fine-grid sweeps, grouped in Table~\ref{tab:rq6-transfer} by what they change. All use $k=5$ except the HotpotQA $k=10$ injection replica and the SciFact, FiQA-2018 and TREC-COVID sweeps. The threshold everywhere is the task's evidence requirement of Section~\ref{sec:protocol}. The rule transfers across encoders, across $k$, across the redundancy mechanism and across datasets. On clean pools, its cost stays within the one-point materiality tolerance everywhere except on the two FiQA-2018 sweeps ($-1.0$ points under injection, significant; $-1.0$ under chunk overlap, not significant) and the two saturated TREC-COVID sweeps ($-1.7$ under injection, significant; $-2.0$ under chunk overlap, not significant on that fifty-query set). Under redundancy, the rule gains on every multi-evidence sweep except on the chunk overlap run of FiQA-2018, where no level differs significantly from $k$-NN. On single-hop NQ-Open the trigger condition is unsatisfiable and the rule is exactly $k$-NN, while on saturated TREC-COVID, whose relevant-set sizes keep the trigger firing on every query, the rule coincides with always-\MMR(0.7): under injection it trades its clean-pool toll for one of the largest heavy-redundancy gains in the table ($+12.9$ points, $95\%$ CI $[10.4,15.8]$, at $\rho=1$).

	All transfer significance statements follow the query-clustered procedure of Appendix~\ref{app:stats}, computed per level and pooled from the displayed runs' per-query rows with Holm adjustment within each sweep. The per-level effects, intervals and trigger rates are released with the analysis outputs.

	\begin{table}[ht]
		\centering
		\small
		\caption{Transfer of the frozen rule ($\tau=\min\{h,k\}$ with gold $h$, an oracle analysis; $D^{*}=\MMR(0.7)$) to further sweeps with no retuning, grouped by what each sweep changes relative to the tuning sweep (\texttt{bge-m3} HotpotQA injection). Each $\Delta$ is the rule's S-Recall@$k$ minus $k$-NN's at the same sweep level, averaged over seeds (positive favors the rule): clean$\Delta$ is the worst $\Delta$ over the near-clean levels ($\rho\le0.05$ on injection sweeps, overlap $\le0.25$ on chunking sweeps), heavy$\Delta$ the $\Delta$ at the heaviest level ($\rho=1$, respectively overlap $0.75$) and pooled$\Delta$ the difference with all levels pooled.}
		\label{tab:rq6-transfer}
		\begin{tabular}{lccc}
			\hline
			\textbf{Transfer target} & \textbf{clean$\Delta$} & \textbf{heavy$\Delta$} & \textbf{pooled$\Delta$} \\
			\hline
			\multicolumn{4}{l}{\emph{Different encoder (injection)}}\\
			\quad Qwen3-Embedding-4B     & $-0.3$ & $+22.5$ & $+9.8$ \\
			\quad all-MiniLM-L6-v2       & $-0.2$ & $+13.6$ & $+5.9$ \\
			\multicolumn{4}{l}{\emph{Different $k$ (injection)}}\\
			\quad HotpotQA ($k{=}10$)   & $\phantom{-}0.0$ & $+24.1$ & $+5.7$ \\
			\multicolumn{4}{l}{\emph{Different mechanism (chunk overlap)}}\\
			\quad HotpotQA      & $\phantom{-}0.0$ & $+1.4$ & $+0.4$ \\
			\multicolumn{4}{l}{\emph{Different dataset (injection)}}\\
			\quad SciFact (single-hop, $k{=}10$) & $-0.7$ & $+1.8$ & $+0.4$ \\
			\quad FiQA-2018 (multi-evidence, $k{=}10$) & $-1.0$ & $+9.2$ & $+2.9$ \\
			\quad 2WikiMultiHopQA (multi-hop) & $\phantom{-}0.0$ & $+10.7$ & $+1.5$ \\
			\quad MuSiQue (multi-hop) & $+0.7$ & $+21.5$ & $+5.6$ \\
			\quad TREC-COVID (saturated, $k{=}10$) & $-1.7$ & $+12.9$ & $+3.1$ \\
			\multicolumn{4}{l}{\emph{Different dataset and mechanism (chunk overlap)}}\\
			\quad NQ-Open (single-hop)   & $\phantom{-}0.0$ & $\phantom{-}0.0$ & $\phantom{-}0.0$ \\
			\quad SciFact (single-hop, $k{=}10$) & $-0.4$ & $+0.4$ & $\phantom{-}0.0$ \\
			\quad FiQA-2018 (multi-evidence, $k{=}10$) & $-1.0$ & $+0.3$ & $-0.2$ \\
			\quad 2WikiMultiHopQA (multi-hop)       & $\phantom{-}0.0$ & $+3.2$ & $+1.0$ \\
			\quad MuSiQue (multi-hop) & $+2.3$ & $+10.0$ & $+5.4$ \\
			\quad TREC-COVID (saturated, $k{=}10$) & $-2.0$ & $+4.1$ & $+0.7$ \\
			\hline
		\end{tabular}
	\end{table}

	\subsection{RQ5: Answer quality under redundancy}
	\label{sec:rq5}

	We re-run the injection sweep on HotpotQA and 2WikiMultiHopQA under the generation protocol of Section~\ref{sec:protocol} with \texttt{flan-t5-base}. We see in Figure~\ref{fig:generation} that the picture established at the retrieval level carries over to answer quality. On clean pools diversification hurts, \MMR costing a significant $3.1$ EM points against $k$-NN (paired risk difference, $95\%$ CI $[-4.0,-2.3]$). Under redundancy, $k$-NN answer quality collapses while the diversifier holds flat, leading $k$-NN significantly from $\rho=0.1$ onward and by $+4.7$ EM points $[3.6,5.8]$ at $\rho=1$. Moreover, as previously observed at the retrieval level, the decision rule tracks the upper envelope, sitting on $k$-NN at $\rho=0$ (firing on under $1\%$ of queries), switching to the diversifier once redundancy appears and staying within $0.2$ EM points of the better of the two policies at every level (Figure~\ref{fig:generation}). 2WikiMultiHopQA tells a softer but similar story: the diversifier's and the rule's advantage materializes only at the highest level ($+2.0$ EM points $[1.2,2.8]$ at $\rho=1$, significantly below $k$-NN on the cleaner levels), the rule staying within $0.7$ EM points of $k$-NN at moderate $\rho$. Because \texttt{flan-t5-base} may handle distractors and redundant context differently from current instruction-tuned models, Table~\ref{tab:generation} repeats the comparison at the clean and heaviest levels with \texttt{Qwen3.8-27B} \cite{qwenteam2026qwen38} under fixed prompts and deterministic decoding. The regime pattern is unchanged: $k$-NN and the rule lead on clean pools, \MMR and the rule under heavy redundancy, the rule again tracking the better fixed policy within $0.1$ points at both endpoints.

	\begin{figure}[ht]
		\centering
		\begin{subfigure}[b]{0.49\linewidth}
			\centering
			\includegraphics[width=\linewidth, alt={Line plot of test exact match against injected redundancy rho on HotpotQA with the flan-t5-base generator for always k-NN, always MMR(0.7) and the rule with tau equal to h. Always k-NN falls from 34 to about 26.5, always MMR stays near 31, and the rule matches k-NN on clean pools and MMR at every injected level.}]{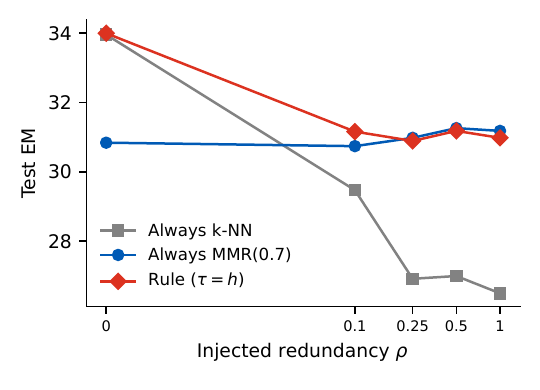}
			\subcaption{HotpotQA}
			\label{fig:generation-hotpot}
		\end{subfigure}\hfill
		\begin{subfigure}[b]{0.49\linewidth}
			\centering
			\includegraphics[width=\linewidth, alt={Line plot of test exact match against injected redundancy rho on 2WikiMultiHopQA with the flan-t5-base generator for always k-NN, always MMR(0.7) and the rule with tau equal to h. Always k-NN stays near 31 up to rho 0.25 and then falls to about 27.4 at rho 1, always MMR stays near 29.5, and the rule tracks k-NN before ending near 29.5.}]{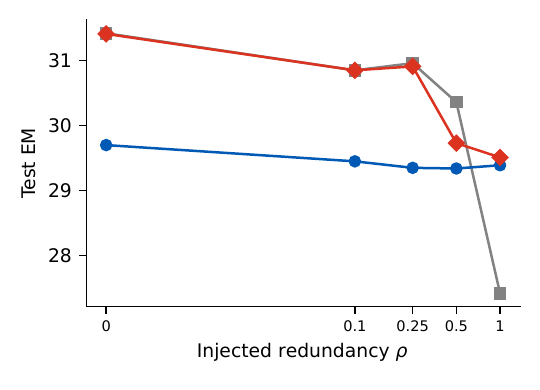}
			\subcaption{2WikiMultiHopQA}
			\label{fig:generation-2wiki}
		\end{subfigure}
		\caption{Exact match vs injected redundancy on (\subref{fig:generation-hotpot}) HotpotQA and (\subref{fig:generation-2wiki}) 2WikiMultiHopQA, generator \texttt{flan-t5-base}.
			Log-scaled $\rho$ axis with ticks at the sweep levels.}
		\label{fig:generation}
	\end{figure}

	\begin{table}[ht]
		\centering
		\small
		\caption{Answer quality (EM/F1) at the clean and heaviest injection levels, generator \texttt{Qwen3.8-27B} (a single run per dataset). Best in bold, second best underlined.}
		\label{tab:generation}
		\setlength{\tabcolsep}{6pt}
		\begin{tabular}{lcc|cc}
			\hline
			& \multicolumn{2}{c|}{HotpotQA} & \multicolumn{2}{c}{2WikiMultiHopQA} \\
			\textbf{Method} & $\rho{=}0$ & $\rho{=}1$ & $\rho{=}0$ & $\rho{=}1$ \\
			\hline
			$k$-NN            & \textbf{50.8}/\textbf{63.6} & 41.7/53.6 & \textbf{63.4}/\underline{71.8} & 49.5/56.5 \\
			\MMR($\lambda{=}0.7$) & 47.6/59.9 & \textbf{47.8}/\textbf{60.5} & 59.3/67.5 & \textbf{55.5}/\textbf{63.5} \\
			rule ($\tau{=}h$)  & \textbf{50.8}/\textbf{63.6} & \textbf{47.8}/\underline{60.4} & \textbf{63.4}/\textbf{71.9} & \underline{55.4}/\underline{63.4} \\
			\hline
		\end{tabular}
	\end{table}

	\section{Discussion}
	\label{sec:discussion}

	\subsection{Reconciling the literature}

	Read jointly, the results of Section~\ref{sec:results} resolve the contradiction of Section~\ref{sec:background}. Studies reporting that diversity benefits retrieval in RAG pipelines operate on redundant pools: chunked long contexts \cite{khan2026dfrag}, iterative pipelines that re-fetch overlapping content \cite{rezaei2025vendirag}, redundancy-heavy benchmarks \cite{corteslemos2025drag,pickett2024infogain} or older first-stage encoders of the fixed-pool selection methods \cite{corteslemos2025drag,li2024smartrag}. AdaGReS makes the dependence to redundancy explicit by targeting near-duplicate chunks \cite{adagres2025}. Studies reporting that diversity does not improve retrieval, including the controlled evaluations of \citet{eibich2024aragog} and \citet{bal2026biomedrag}, operate on nonredundant pools. The cross-encoder comparison of Section~\ref{sec:rq1} likewise reproduces the finding that reranking, not diversification, is the high-leverage stage of the pipeline \cite{wang2024bestpractices}. Neither camp is wrong about its own regime: the error is the implicit claim that the conclusion transfers across regimes.

	The reported sweeps quantify how low the boundary between regimes sits on the strong-encoder multi-hop tasks. A measured near-duplicate pair fraction of the order of $10^{-3}$ is enough for $k$-NN retrieval to collapse and under the weaker encoder the sign change appears earlier still (descriptively, see Figure~\ref{fig:crossover-minilm}), in line with evaluations that use older encoders landing in the regime where diversification pays. The oracle bounds explain the adaptive end of the literature the same way: adaptivity, like diversification itself, is valuable in proportion to measured redundancy, the per-query headroom moving from about $6\%$ of queries to over half across the crossover (Figure~\ref{fig:oracle}), which is why per-query selection of the diversity level helps in the chunked long-context setting of \citet{khan2026dfrag} yet has nothing to capture on many benchmarks.

	Other works in the literature reporting positive results regarding retrieval diversification concern notions of diversity different from the one studied here: diversity on the query side, whether through lexical decomposition \cite{zhang2025drag} or coverage of a sub-query demand distribution \cite{zhang2026georag}, selection for long contexts \cite{wang2025diversityllm} and open-ended tasks in which answer diversity is itself rewarded \cite{salemi2025planrefine,diverge2026}. None of these evaluates geometric diversification of a fixed candidate pool where diversity is merely instrumental, so their conclusions complement rather than contradict the present analysis.

	\subsection{Clean pools resist diversification}

	The systematic dissociation between annotation-free and intent-aware diversity metrics is, after the decision rule, the most practically consequential observation of the study. APD and the Vendi score are routinely reported as evidence that a diversified selection is better. However, on every dataset we evaluate, methods that maximize them minimize $\alpha$-NDCG, S-Recall, and answer quality. Our interpretation is geometric: in modern embedding spaces, the distinct facets of a multi-hop question occupy a compact neighborhood of the query, so pushing a selection out of that region substitutes gold passages for related but irrelevant ones. These are harmful for the generator \cite{cuconasu2024noise,amiraz2025distracting} and add nothing to the context sufficiency that generation actually requires \cite{joren2025sufficient}. ArguAna is the extreme case: the single relevant counter-argument sits directly adjacent to the query in embedding space, so any outward pressure discards it, which is how \MMR collapses to an NDCG@10 of $0.2$ (Table~\ref{tab:rq1-beir}). The all-methods oracle bound of Section~\ref{sec:rq3} allows us to push the analysis further: on nonredundant pools the problem is not that the diversifiers are badly tuned or insufficiently adaptive, but that there is almost nothing for diversification to gain over $k$-NN. Thus, evaluations should treat geometric spread as a descriptive statistic rather than an objective and should instead report coverage-aware metrics whenever labels permit. Those spread statistics may nonetheless be used as diagnostics, analogously to the decision rule using the Vendi score as a redundancy gauge.

	\subsection{Reading the regime from the \RNGSCORE}

	The margin trajectories obtained from the \RNGSCORE in Table~\ref{tab:rq2-margins} are the hinge geometry at work. A negative margin restricts the obstruction penalty to pairs at duplicate distance scale, demoting copies while leaving distinct evidence untouched, while a positive margin penalizes more severely distinct documents, meaning that there is no possible gain to be made with near-duplicates. The margin thus reads duplicate-scale structure specifically, not the regime at large, and the other trajectories confirm it: a sign flip occurs at different levels for different encoders because each geometry places duplicate pairs at a distinct scale relative to query--document distances. The weak encoder stays positive longest, its low-redundancy gains coming from a more conflated structure rather than from literal duplicates. Furthermore, seeds disagree in sign exactly where the validation objective is flat, i.e. in the transition zones. On SciFact however, they disagree across nearly the whole axis, meaning that the margin is reading noise precisely where there is nothing to measure. Therefore, we see that a small negative margin signals duplicate structure worth removing, while a fallback-region or seed-unstable margin indicates a pool with nothing to remove.

	\subsection{Deploying the rule}

	The sweeps show that no fixed method dominates the redundancy axis, so a deployed policy should adapt accordingly. The decision rule proposed here adapts on the right variable: its trigger reads off exactly the quantity whose collapse causes the $k$-NN failure and the per-query distributions of $T(q)$ barely overlap across the crossover, which is why a wide band of thresholds performs identically and why the frozen rule transfers to redundancy levels and mechanisms it never saw. Adapting per query matters because the crossover location itself is specific to each task. The steepness of the turn is consistent with evidence multiplicity: gold-concentrated redundancy can crowd out a needed document only when the query needs more than one, in line with which the multi-hop datasets and multi-evidence FiQA collapse at a few near-duplicate per mille while single-hop SciFact and NQ-Open reach the regime only under extreme duplication, mildly even then. Across our experiments, the observed boundary spans from below $0.0006$ (the weak encoder) to approximately $0.006$ (FiQA at $k=10$, where the gains become stable; its clean-pool effect never turns significantly negative, so there is no sign change to bracket there) and does not exist on SciFact. Thus, although the specific location does not generalize, the existence and sharpness of a measurable crossover does on multi-hop datasets and the trigger localizes that boundary per query without ever measuring it.

	Evidence multiplicity also anchors the rule's threshold. The two variables of the analysis play the distinct roles that the rule's architecture mirrors: redundancy is the lever, the dynamic variable that is usually unmeasured in production pipelines; while evidence multiplicity supplies the threshold, treated in this paper as a static workload property because our evaluation reads $h$ from gold annotations, the hop count on multi-hop QA and the relevant-set size on the graded-relevance tasks, quantities that a deployment does not have per query. Label-free per-query estimators nevertheless exist off the shelf, such as query-complexity classifiers routing queries by predicted retrieval depth \cite{jeong2024adaptiverag} and query decomposition counting a question's sub-questions directly \cite{press2023selfask}. Thus, the practical cost of fixing $\tau=h$ is a single integer either estimated per query or for a whole workload. A natural refinement, left to future work, is therefore to supply the threshold with the query's own estimated evidence requirement, $T(q)<\hat{h}(q)$ with $\hat{h}(q)$ obtained from one of these estimators. This would be another step forward since, to our knowledge, no existing adaptive pipeline couples a diversification decision with an estimated hop count, the adjacent adaptive mechanisms selecting instead how many documents to keep per query \cite{meng2025dps,xu2025car} or diversifying across the inferred interpretations of ambiguous questions \cite{in2025diva}.

	The rule is the instrumental-credit counterpart of selective diversification \cite{santos2010selective}, which gates web-search diversification on predicted query ambiguity \cite{cronen2002predicting}. Among the adaptive pipelines of \citet{khan2026dfrag}, \citet{rezaei2025vendirag} and \citet{stouras2026portfolios}, it is the minimal member, spending an eigendecomposition per query where those spend planner and evaluator LLM calls or a trained router. On clean pools, the all-methods oracle bound implies that no router over the evaluated method--parameter library can do appreciably better on these test instances, and on redundant pools the rule's remaining gap to the oracle bounds what richer routing policies over the same library can still contribute.

	\subsection{A reachable fallback as a design requirement}

	The downside analysis generalizes beyond the methods compared. A fixed diversification strength interacts unpredictably with pool geometry, as \MMR's collapse on ArguAna illustrates (Table~\ref{tab:rq1-beir}), so a reachable relevance-only fallback should be a design requirement for any diversifier, met structurally (an operating point that provably reproduces $k$-NN, as at the \RNGSCORE's margin extremes or \MMR's $\lambda=1$) or operationally (gating the diversifier behind a redundancy trigger). The two routes are complementary, the gated rule achieving the better headline numbers and the exact-fallback score supplying an interpretable margin.

	\section{Limitations}
	\label{sec:limitations}

	Six limitations qualify the conclusions. None, in our assessment, threatens the regime analysis itself.
	\begin{enumerate*}[label=(\roman{enumi})]
		\item \emph{Statistical coverage.} The sweep, transfer and generation analyses carry the query-clustered paired inference of Section~\ref{sec:protocol} and Appendix~\ref{app:stats}, whereas a few result families are reported descriptively: the \texttt{Qwen3.8-27B} generation columns of Tables~\ref{tab:rq1-hotpot} and~\ref{tab:generation} (single deterministic runs, per-query answers released, no inferential claim) and the families with only aggregate evidence (the two argument-retrieval BEIR tasks, the non-HotpotQA cross-encoder datasets and the NQ-Open chunking sweep). Conclusions on the $50$--$300$-query BEIR sets are exploratory, their attainable precision being about two points, and the fifteen transfer sweeps share queries and corpora rather than forming independent validations, so that dataset-level generalization would require a hierarchical analysis with the dataset as the higher-level unit or more independently sampled corpora.

		\item \emph{Redundancy realism.} Redundancy was simulated using sentence-shuffled copies and natural chunk-overlap protocols, but other mechanisms such as injected paraphrases remain untested, as do an independently sampled production corpus and the sensitivity of the crossover location to the injection protocol. The chunking sweeps also score a subtopic as covered when any window of a gold passage is selected, an optimistic convention (a sixty-word window need not contain the answer-bearing content) that the chunking protocol does not pair with generation.

		\item \emph{Scope of the rule and its guarantees.} The rule's threshold is the task's evidence requirement, which in deployment requires task-specific knowledge of evidence multiplicity rather than per-query labels, and the rule's behavior under estimated rather than annotated $h$ is untested. The $k=10$ transfer evidence is limited to the HotpotQA injection replica and the SciFact, FiQA-2018 and TREC-COVID sweeps. Margin tuning inherits validation-split noise, so an adequate tuning set is required (on TREC-COVID, which only contains fifty queries, the selected margin varies across splits). Moreover, the oracle applies only to policies that adapt a scalar or the method per query, not policies that change the pool or the query.

		\item \emph{Task and domain scope.} The study covers English factoid, multi-hop QA and BEIR-style retrieval: multilingual retrieval, summarization and exploratory search are untested. Open-ended information seeking, where comprehensiveness or answer diversity gives diversity independent value \cite{salemi2025planrefine,diverge2026}, is out of scope and our conclusions do not extend to it.

		\item \emph{Generation.} The generation columns of Table~\ref{tab:rq1-hotpot} and the regime pattern of Section~\ref{sec:rq5}, confirmed at the clean and heaviest levels, use the instruction-tuned \texttt{Qwen3.8-27B}, but the intermediate injection levels and the remaining datasets rest on \texttt{flan-t5-base} alone, and prompt-format sensitivity is untested for both generators, as is the split and injection-draw variability of the single runs.

		\item \emph{Baselines.} The diversifiers compared are the classical fixed-pressure methods: isolating recent adaptive and Vendi-objective selectors on the same fixed pools, and a quality-weighted \DPP\ with a tunable trade-off, is left to future work.
	\end{enumerate*}

	\section{Conclusion}
	\label{sec:conclusion}

	This paper asked when the retrieval stage of a RAG pipeline should diversify and answered with a regime analysis organized around a main observable variable: whether the candidate pool contains redundancy. On clean pools, the answer on the task families studied here is: essentially never. In most cases, the fixed-pressure diversifiers degrade relevance, coverage-aware diversity and answer quality, significantly so on the primary coverage endpoint. The geometric spread that diversifiers deliver is anti-correlated with the coverage they promise and the all-methods per-query oracle certifies that only about $6\%$ of queries admit any gain from the evaluated diversifiers at all. The regime changes when the pool itself is redundant, whether by injection or by overlapping chunking. Then, on the multi-evidence tasks, the relevance-aware diversifiers overtake nearest-neighbor selection, by a wide margin under injection and by a smaller one under chunk overlap.

	The practical contribution is that the regime need not be guessed. On the multi-hop tasks, the crossover occurs at a measured near-duplicate pair fraction far below what casual inspection would flag, bracketed on the order of $10^{-3}$ under both strong encoders, although the location itself is specific to each dataset and encoder. In addition, a label-free statistic that any retrieval system can compute, the Vendi score of its $k$ nearest neighbors, detects the regime per query, so the crossover location is not needed to deploy the rule. The proposed decision rule, i.e. diversify when that score falls below the number of evidence pieces $h$ the query requires, keeps its clean-pool cost within the one-point materiality tolerance except on FiQA-2018 and saturated TREC-COVID, where large relevant sets keep the trigger firing even on clean pools, recovers most of the value of perfect regime knowledge on redundant ones and transfers with no retuning across encoders, datasets and $k$. When the threshold is left free, validation recovers a large plateau of possible values which includes $h$, namely the evidence multiplicity, which is a property of the workload a deployment knows or should estimate in practice.

	The analysis also yields an evaluation standard and an instrument: misspecification downside, which across the classical diversifiers and the $k$-NN default spans $10$--$33$ coverage points, and the \RNGSCORE, whose two provable nearest-neighbor limits anchor its tuning to the relevance-only baseline and whose tuned margin turns negative once the pool acquires duplicate structure.

	\paragraph{Code and data availability.}
	The full evaluation pipeline (sweep scripts, derived analyses, decision-rule and transfer tooling, and the per-run result files behind every table and figure) is available on GitHub\footnote{ \href{https://github.com/GuillaumeBrouillette/finding-the-right-balance}{https://github.com/GuillaumeBrouillette/finding-the-right-balance}}.

	\section*{Acknowledgements}
	The authors gratefully acknowledge the Universit\'e du Qu\'ebec \`a Trois-Rivi\`eres (UQTR), which funds the Laboratoire d'intelligence artificielle appliqu\'ee (LI2A). Guillaume Brouillette was supported by the LI2A, and Guillaume Brouillette and Faustin Kagabo were supported by Mitacs through the Mitacs Accelerate program.

	{
		\bibliographystyle{plainnat}
		\bibliography{references}}

	\appendix

	\section{Theoretical properties}
	\label{app:theory}

	We present here the facts summarized in Section~\ref{sec:theory}. Recall that for a candidate pool $V\subset\R^D$, a candidate $w\in V$, a query $q\in\R^D$ and a parameter $\gamma\in\R$, the \RNGSCORE is defined as
	\begin{equation*}
		\score{\gamma}(w;q,V)
		= d(q,w) + \sum_{v\in P(w;q,V)} \hinge{d(q,w)-d(v,w)+\gamma},
	\end{equation*}
	where $P(w;q,V)=\{v\in V : d(q,v) < d(q,w)\}$ is the set of closer competitors of $w$.

	\begin{prop}[Exact recovery of the hard admissibility baseline]
		\label{prop:baseline}
		For any $w\in V$ and any $\gamma\in\R$, $\score{\gamma}(w;q,V)=d(q,w)$ if and only if every closer competitor $v\in P(w;q,V)$ satisfies $d(v,w)\ge d(q,w)+\gamma$.
	\end{prop}
	\begin{proof}
		Every summand in the \RNGSCORE is nonnegative, so $\score{\gamma}(w;q,V)=d(q,w)$ exactly when every hinge term is zero. The hinge for $v$ is zero iff $d(q,w)-d(v,w)+\gamma\le 0$, i.e.\ $d(v,w)\ge d(q,w)+\gamma$.
	\end{proof}

	\begin{prop}[Monotonicity in $\gamma$]
		\label{prop:monotone}
		For every fixed $w\in V$, the map $\gamma\mapsto \score{\gamma}(w;q,V)$ is nondecreasing.
	\end{prop}
	\begin{proof}
		Each term $\hinge{d(q,w)-d(v,w)+\gamma}$ is nondecreasing in $\gamma$ and a finite sum of nondecreasing functions is nondecreasing.
	\end{proof}

	The monotonicity is pointwise in $w$: different candidates' scores can grow at different rates as $\gamma$ increases, so the proposition implies nothing about the induced ranking, the geometric diversity of the selected set or downstream task utility, none of which need vary monotonically with $\gamma$.

	\begin{prop}[Similarity invariance]
		\label{prop:similarity}
		Let $d$ and $\tilde d$ be symmetric nonnegative dissimilarities on $X$ and $\tilde X$ respectively, let $V\subset X$ be finite with $q\in X$, and suppose $f:X\to\tilde X$ is injective on $\{q\}\cup V$ and satisfies $\tilde d(f(x),f(y))=\lambda\,d(x,y)$ for all $x,y\in\{q\}\cup V$ and some $\lambda>0$. Then, writing $\score{\gamma}^d$ for the \RNGSCORE under $d$,
		\[
		\score{\lambda\gamma}^{\tilde d}(f(w);f(q),f(V)) = \lambda\, \score{\gamma}^d(w;q,V)
		\]
		for every $w\in V$. Hence, isometries preserve the score while similarities preserve the induced ordering, provided a deterministic tie rule is applied consistently on both sides.
	\end{prop}
	\begin{proof}
		Injectivity of $f$ on $\{q\}\cup V$ ensures that $f(V)$ is a pool of the same cardinality as $V$, meaning that distinct candidates cannot collapse in $\tilde X$. Since $\lambda>0$, we see that $d(q,v)<d(q,w)$ iff $\tilde d(f(q),f(v))<\tilde d(f(q),f(w))$, so the closer-competitor set of $f(w)$ is exactly $f(P(w;q,V))$. Also, every term in the score is scaled by $\lambda$ and since $\hinge{\lambda x}=\lambda\hinge{x}$, we deduce that $\score{\lambda\gamma}^{\tilde d}(f(w);f(q),f(V))$ is simply a scaling of $\score{\gamma}^d(w;q,V)$ by factor $\lambda$.

		Note that both the distances and the margin are scaled by the same factor $\lambda$: scaling only part of the construction does not preserve the score.
	\end{proof}

	\begin{prop}[Nearest-neighbor limit for $\gamma\to-\infty$]
		\label{prop:knnlimit}
		Let $\mathcal{A}(q,V)=\{(w,v): w\in V,\ v\in P(w;q,V)\}$ denote the set of active pairs. If $\mathcal{A}(q,V)=\varnothing$, then $\score{\gamma}(w;q,V)=d(q,w)$ for every $w\in V$ and every $\gamma\in\R$. Otherwise, let
		\[
		\gamma_0(q,V)=\min_{(w,v)\in\mathcal{A}(q,V)}\{d(v,w)-d(q,w)\}.
		\]
		Then for every $\gamma\le\gamma_0(q,V)$, we have $\score{\gamma}(w;q,V)=d(q,w)$ for all $w\in V$, and ranking by \RNGSCORE coincides with ranking by distance to the query.
	\end{prop}
	\begin{proof}
		If $\mathcal{A}(q,V)$ is empty, every sum in \eqref{eq:rngscore} is empty. Otherwise, for every $(w,v)\in\mathcal{A}(q,V)$ and $\gamma\le\gamma_0(q,V)$,
		\[
		d(q,w)-d(v,w)+\gamma \le d(q,w)-d(v,w)+\gamma_0(q,V) \le 0,
		\]
		so every hinge term vanishes simultaneously.
	\end{proof}

	For unit-normalized embeddings under cosine distance, $d\in[0,2]$ and hence $d(v,w)-d(q,w)\ge-2$, so the uniform choice $\gamma\le-2$ realizes the fallback on every finite pool. A smaller-magnitude grid endpoint such as $\gamma=-0.5$ carries no such guarantee: where the paper reports that it reproduces $k$-NN, this is an empirical property verified on the pools in question (Section~\ref{sec:protocol}).

	\begin{prop}[Nearest-neighbor limit for $\gamma\to+\infty$]
		\label{prop:knnlimit_pos}
		Suppose $V$ is finite and in general position with respect to $q$, meaning that the values $\{d(q,v):v\in V\}$ are pairwise distinct. For $w\in V$, write $r_w=d(q,w)$, $n_w=|P(w;q,V)|$ and
		\[
		b_w=(n_w+1)\,r_w-\sum_{v\in P(w;q,V)}d(v,w),
		\]
		and define, with the convention $\max\varnothing=-\infty$,
		\[
		\gamma_{\mathrm{act}}=\max_{(w,v)\in\mathcal{A}(q,V)}\{d(v,w)-r_w\},
		\qquad
		\gamma_{\mathrm{ord}}=\max_{\substack{u,w\in V\\ n_u<n_w}}\frac{b_u-b_w}{n_w-n_u}.
		\]
		For every $\gamma>\max\{\gamma_{\mathrm{act}},\gamma_{\mathrm{ord}}\}$, the ordering induced by $\score{\gamma}$ coincides with the ordering induced by $d(q,\cdot)$, i.e.\ with the $k$-NN ranking.
	\end{prop}
	\begin{proof}
		For $\gamma\ge\gamma_{\mathrm{act}}$, every hinge term is active, so
		\[
		\score{\gamma}(w;q,V)=r_w+\sum_{v\in P(w;q,V)}\bigl(r_w-d(v,w)+\gamma\bigr)=n_w\,\gamma+b_w.
		\]
		Since the query distances are pairwise distinct, $n_w$ equals the query-distance rank of $w$ minus one, so $w\mapsto n_w$ is a bijection onto $\{0,\dots,|V|-1\}$. If $n_u<n_w$, then $\gamma>\gamma_{\mathrm{ord}}$ gives $(n_w-n_u)\,\gamma>b_u-b_w$ and hence $n_u\gamma+b_u<n_w\gamma+b_w$. The score ordering therefore agrees with the strictly increasing values $n_w$, which agree with the query-distance ordering.
	\end{proof}

	\begin{cor}[Fallback under validated tuning]
		\label{cor:fallback}
		Let $A\subset\R$ be a finite tuning grid containing some margin $\gamma$ that induces the $k$-NN ranking on every validation pool (for instance $\gamma\le\gamma_0(q,V)$ for every validation query, or uniformly $\gamma\le-2$ under cosine distance, Proposition~\ref{prop:knnlimit}), and let $\gamma^{*}$ maximize any evaluation functional of the induced top-$k$ ranking over the validation set. Then the validation performance at $\gamma^{*}$ is at least that of nearest-neighbor ranking, which is realised by a grid member.
	\end{cor}

	Corollary~\ref{cor:fallback} is an optimization identity on the validation sample, not a statement about deployment queries. The following bound is its finite-sample population counterpart.

	\begin{thm}[Finite-validation fallback bound]
		\label{thm:valbound}
		Let $Z_1,\ldots,Z_n$ be independent validation units drawn from the deployment distribution, let $0<\delta<1$ and let $A$ be a nonempty finite margin grid fixed independently of these units. Suppose the per-unit utility $u(\gamma,Z)$ takes values in $[0,1]$ and that some $\gamma_{\mathrm{nn}}\in A$ induces the $k$-NN ranking throughout the deployment support (under cosine distance, any $\gamma_{\mathrm{nn}}\le-2$ suffices by Proposition~\ref{prop:knnlimit}). Define
		\[
		U(\gamma)=\E\,u(\gamma,Z),\qquad
		\widehat U_n(\gamma)=\frac1n\sum_{i=1}^n u(\gamma,Z_i),\qquad
		\widehat\gamma\in\arg\max_{\gamma\in A}\widehat U_n(\gamma).
		\]
		Then, with probability at least $1-\delta$,
		\[
		U(\widehat\gamma)\;\ge\;U(\gamma_{\mathrm{nn}})-2\sqrt{\frac{\log(2|A|/\delta)}{2n}}.
		\]
	\end{thm}
	\begin{proof}
		By Hoeffding's inequality and a union bound over $A$, with probability at least $1-\delta$,
		\[
		\sup_{\gamma\in A}\bigl|\widehat U_n(\gamma)-U(\gamma)\bigr|\le\varepsilon,
		\qquad
		\varepsilon=\sqrt{\frac{\log(2|A|/\delta)}{2n}}.
		\]
		On this event, $U(\widehat\gamma)\ge\widehat U_n(\widehat\gamma)-\varepsilon\ge\widehat U_n(\gamma_{\mathrm{nn}})-\varepsilon\ge U(\gamma_{\mathrm{nn}})-2\varepsilon$.
	\end{proof}

	For a utility reported on a $0$--$100$ scale, the error term is multiplied by $100$. The bound requires independent validation units: when queries share source documents or construction families, $n$ counts independent clusters, each contributing one bounded cluster-level utility. It does not cover distribution shift, grids adjusted after inspecting test results, or a grid endpoint that merely happened to reproduce $k$-NN on the observed validation pools.

	\section{Cross-encoder dissimilarity}
	\label{app:ce-variants}

	This appendix records how the \RNGSCORE, defined through an embedding-space distance, is adapted to the cross-encoder setting of Section~\ref{sec:rq1-crossencoder}.

	A \emph{symmetric nonnegative dissimilarity} on a set $X$ is a function $d:X\times X\to[0,\infty)$ satisfying, for all $x,y\in X$, $d(x,x)=0$ and $d(x,y)=d(y,x)$. Unlike a metric, it does not require $d(x,y)>0$ for $x\neq y$ nor the triangle inequality $d(x,z)\leq d(x,y)+d(y,z)$. The weaker requirement is what the construction below needs: a cross-encoder may assign a distinct candidate relevance one, in which case $d_{\CE}(q,v)=0$ and separation fails.

	In what follows, $\similarity{\CE}(q,w)\in[0,1]$ denotes the cross-encoder relevance score of candidate $w$ for query $q$, and $d$ denotes the underlying embedding-space distance (used for document--document comparisons and for the pool radius below).

	Essentially, the construction defines a symmetric nonnegative dissimilarity $d_{\CE}$ on $V\cup\{q\}$ and applies the \RNGSCORE directly: set $d_{\CE}(q,q)=0$ and $d_{\CE}(q,v)=d_{\CE}(v,q)=1-\similarity{\CE}(q,v)$ for all $v\in V$, so that
	\[
	P(w;q,V)=\bigl\{v\in V : \similarity{\CE}(q,v) > \similarity{\CE}(q,w)\bigr\},
	\]
	i.e.\ the closer competitors of $w$ are exactly the candidates the cross-encoder ranks above it. This is an instance of the decoupled-geometry observation of Section~\ref{sec:kernel}: relevance ordering from the cross-encoder, redundancy from the embedding space. The document--document side requires making the two scales commensurable per pool. The reference choice is the linear normalization
	\[
	d_{\CE}(v,w) = \frac{d(v,w)}{R},
	\qquad R=\max_{v\in V} d(q,v),
	\]
	giving
	\[
	\score{\gamma}(w;q,V) = 1-\similarity{\CE}(q,w) + \sum_{v\in P(w;q,V)}\hinge{1-\similarity{\CE}(q,w)-\tfrac{d(v,w)}{R}+\gamma}.
	\]
	Here $R$ is the pool radius in embedding space, positive unless every candidate coincides with the query, in which case the implementation sets $R=1$. Note that the maximal query--document distance need not bound the document--document distances, so $d_{\CE}(v,w)$ may exceed one. To obtain a bounded normalization, we could replace $R$ with $R_{\mathrm{all}}$, the largest embedding distance over all pairs of $V\cup\{q\}$. Nevertheless, in the experiments, four alternative embedding distance variants (sigmoidal and reciprocal normalizations of $d(v,w)$, a relevance-scaled distance and the raw cosine distance) were considered as well. All five variants land within half an F1 point of one another at their tuned margins, which suggests that the normalization choice does not have a substantial impact on the results.

	In the cross-encoder experiments of Section~\ref{sec:rq1-crossencoder}, the validation-tuned margins land at or near the fallback limits on every dataset: $\gamma^{*}=-0.5$ on HotpotQA, $-0.3$ on 2WikiMultiHopQA (where Proposition~\ref{prop:knnlimit} deactivates the penalty and the cross-encoder ranking is reproduced exactly), $-0.35$ on NQ-Open and $+1.1$ on MuSiQue, near the positive limit of Proposition~\ref{prop:knnlimit_pos}.

	\section{Kernelization}
	\label{app:kernel}

	For any positive-semidefinite kernel $K$ on $V\cup\{q\}$, the kernel-induced squared distance $d_K(x,y)=K(x,x)+K(y,y)-2K(x,y)$ is a symmetric nonnegative dissimilarity, which is all the \RNGSCORE requires, and every hinge term depends on kernel evaluations only so the score never needs explicit coordinates. Since a \DPP\ consumes the same input, one kernel feeds both methods: \DPP\ performs global log-determinant inference, the \RNGSCORE local obstruction scoring. The binary lune condition is invariant under squaring ($x<y\Leftrightarrow x^{2}<y^{2}$ for nonnegative reals), so kernelization changes only the soft penalty weighting, not the obstruction graph. For unit-normalized embeddings, the linear-kernel $d_K$ is twice the cosine distance, so the kernel view adds content exactly where vector norms carry information.

	\section{Statistical procedure}
	\label{app:stats}

	\paragraph{Design.}
	The primary endpoint is S-Recall@$k$. For every query, seed and condition of a sweep, the analysis forms the paired difference between each method's validation-selected member and $k$-NN on the same pool, then averages seed observations within each query: seeds are not independent replications, their $80\%$ test splits sharing most queries. Uncertainty comes from a query-clustered bootstrap ($1{,}999$ resamples, fixed seed) that resamples whole query trajectories, one weight vector shared across all methods and conditions of a run, so the paired and repeated-measures structure survives every resample. Level effects carry two-sided $95\%$ percentile intervals and centered bootstrap $p$-values. The prespecified primary hypothesis of each sweep is the method$\times$redundancy interaction --- the heavy-minus-clean change of a method's paired effect --- tested one-sidedly for improvement with redundancy. Holm step-down adjustment controls the familywise error separately over a sweep's method-by-level effects and over its primary interactions: ``significant'' refers to the adjusted value, raw values being retained in the released archive. Each family comprises six selectors, the four of Table~\ref{tab:interactions} and two auxiliary ones the tables do not display (a near-duplicate filter and a greedy Vendi-score selector), which makes the adjustment conservative for the reported four. The re-executions ran all six on every sweep; the oracle analyses use the originally executed runs, whose library is the five displayed selectors on the HotpotQA, SciFact and NQ-Open sweeps plus the two auxiliary ones elsewhere, which can only raise the bound.

	\paragraph{Primary interaction results.}
	Table~\ref{tab:interactions} reports the interaction for the four evaluated diversifiers on all thirteen sweeps. Every entry is positive and all are significant after Holm adjustment except \MMR$^{*}$ under SciFact injection: on single-hop SciFact diversification has nothing to recover, so the \MMR{} effect does not change with redundancy, consistent with the evidence-multiplicity reading. Whether evidence multiplicity moderates the crossover is not separately identifiable from dataset effects, single- and multi-evidence queries occurring largely in different datasets: the MuSiQue hop-count analysis of Appendix~\ref{app:single-stage} is the within-dataset evidence.

	\begin{table}[ht]
		\centering
		\small
		\caption{Primary test: heavy-minus-clean change of each method's paired S-Recall effect against $k$-NN, in percentage points (one-sided $95\%$ lower bound in parentheses). All entries are significant after Holm adjustment except the one marked $\dagger$.}
		\label{tab:interactions}
		\begin{tabular}{llcccc}
			\hline
			\textbf{Dataset} & \textbf{Mechanism} & \MMR$^{*}$ & Maxmin & Greedy \DPP & \RNGSCORE \\
			\hline
			HotpotQA & injection & $+23.0$ {\scriptsize($+22.5$)} & $+32.4$ {\scriptsize($+31.8$)} & $+29.5$ {\scriptsize($+28.9$)} & $+12.5$ {\scriptsize($+12.1$)} \\
			HotpotQA & chunk overlap & $+2.1$ {\scriptsize($+1.8$)} & $+4.8$ {\scriptsize($+4.4$)} & $+4.4$ {\scriptsize($+4.0$)} & $+3.0$ {\scriptsize($+2.7$)} \\
			2WikiMultiHopQA & injection & $+22.0$ {\scriptsize($+21.7$)} & $+27.4$ {\scriptsize($+26.8$)} & $+24.3$ {\scriptsize($+23.9$)} & $+12.1$ {\scriptsize($+11.8$)} \\
			2WikiMultiHopQA & chunk overlap & $+4.3$ {\scriptsize($+4.0$)} & $+9.9$ {\scriptsize($+9.5$)} & $+8.3$ {\scriptsize($+7.9$)} & $+6.0$ {\scriptsize($+5.7$)} \\
			MuSiQue & injection & $+23.8$ {\scriptsize($+23.0$)} & $+31.6$ {\scriptsize($+30.3$)} & $+27.5$ {\scriptsize($+26.6$)} & $+17.3$ {\scriptsize($+16.5$)} \\
			MuSiQue & chunk overlap & $+9.6$ {\scriptsize($+8.7$)} & $+13.5$ {\scriptsize($+12.2$)} & $+10.8$ {\scriptsize($+9.9$)} & $+9.9$ {\scriptsize($+9.1$)} \\
			NQ-Open & chunk overlap & $+2.6$ {\scriptsize($+2.1$)} & $+3.2$ {\scriptsize($+2.5$)} & $+2.8$ {\scriptsize($+2.2$)} & $+2.4$ {\scriptsize($+1.9$)} \\
			SciFact & injection & $+0.5^{\dagger}$ {\scriptsize($-1.3$)} & $+30.9$ {\scriptsize($+26.4$)} & $+13.3$ {\scriptsize($+9.4$)} & $+2.0$ {\scriptsize($+0.9$)} \\
			SciFact & chunk overlap & $+3.1$ {\scriptsize($+0.9$)} & $+8.4$ {\scriptsize($+4.7$)} & $+4.8$ {\scriptsize($+2.0$)} & $+4.6$ {\scriptsize($+2.7$)} \\
			FiQA-2018 & injection & $+10.3$ {\scriptsize($+8.8$)} & $+34.1$ {\scriptsize($+30.8$)} & $+27.2$ {\scriptsize($+24.3$)} & $+10.5$ {\scriptsize($+9.1$)} \\
			FiQA-2018 & chunk overlap & $+4.3$ {\scriptsize($+2.5$)} & $+6.3$ {\scriptsize($+3.5$)} & $+7.3$ {\scriptsize($+4.6$)} & $+3.4$ {\scriptsize($+2.1$)} \\
			TREC-COVID & injection & $+12.8$ {\scriptsize($+10.6$)} & $+11.5$ {\scriptsize($+8.4$)} & $+16.7$ {\scriptsize($+13.8$)} & $+14.2$ {\scriptsize($+11.9$)} \\
			TREC-COVID & chunk overlap & $+4.4$ {\scriptsize($+3.2$)} & $+5.6$ {\scriptsize($+2.5$)} & $+6.0$ {\scriptsize($+4.1$)} & $+5.1$ {\scriptsize($+3.6$)} \\
			\hline
		\end{tabular}
	\end{table}

	\paragraph{Provenance and coverage.}
	The method-effect and interaction values are computed from the fixed-pool re-executions of the thirteen sweeps (Section~\ref{sec:protocol}). The displayed injection and chunk-overlap sweep tables (Tables~\ref{tab:rq4-redundancy} and~\ref{tab:rq4-chunking}) likewise show the fixed-pool re-executions, as query-level test means with seed observations averaged within query. The remaining tables, and the figures, show the originally executed runs, whose paired effects agree with the corrected ones to about a tenth of a point on the headline sweep. Level values shift more on the small-pool sweeps, where freezing the candidate count shrinks the heavy pools: $k$-NN moves by up to thirteen points on 2WikiMultiHopQA and MuSiQue (already reflected in Table~\ref{tab:rq4-chunking}) while the paired interactions stay positive and significant (Table~\ref{tab:interactions}). The displays on original pools most affected are therefore the chunk-overlap transfer rows and oracle panel, Table~\ref{tab:single-stage} (whose $k$-NN levels of $85.7$ and $69.2$ read $88.5$ and $74.1$ under fixed pools) and the 2WikiMultiHopQA and MuSiQue injection rows of Table~\ref{tab:rq6-transfer}, where $k$-NN at $\rho=1$ is nine to ten points lower on the original pools, so their heavy$\Delta$ are upper readings (the fixed-pool interactions of the tuned \MMR\ remain $+22.0$ and $+23.8$ points). Figure~\ref{fig:generation} also keeps the original pools. The \texttt{Qwen3.8-27B} clean- and heavy-level runs of Table~\ref{tab:generation} were likewise executed under the original truncation policy of Section~\ref{sec:protocol}; their $\rho=0$ cells are unaffected, since no copies are injected there. The transfer, oracle, regret and generation quantities are computed under the same procedure from the per-query rows of the originally executed runs, exact match as a paired risk difference since the outcome is binary. The \texttt{Qwen3.8-27B} generation columns, which are single deterministic runs, and the families with only aggregate retention -- ArguAna and Webis-Touch\'e2020, and the non-HotpotQA cross-encoder comparisons -- are reported descriptively. Embedding tensors, post-encoder rankings and the \texttt{flan-t5-base} answer strings were not serialized by the original runs; they can be regenerated from the frozen pools, the pinned model revisions and the archived code. The complete archive (contrasts, intervals, raw and adjusted $p$-values, improved/tied/harmed fractions, and an independent validator) is released with the code.

	\paragraph{Limits.}
	Queries sharing source documents carry no family identifiers, so the query is the finest resampling unit and the intervals do not account for residual between-query dependence. The realized interval widths imply minimum detectable effects of roughly $0.4$ points on the HotpotQA-scale sets and $1.5$--$2.5$ points on the $50$--$300$-query BEIR sets, which is why conclusions on the small sets are exploratory. The materiality tolerance is a reporting convention, not an equivalence test.

	\section{Single-stage comparisons on the remaining QA datasets}
	\label{app:single-stage}

	Table~\ref{tab:single-stage} reports single-stage reranking on the non-overlapping pools of the three QA datasets whose main-body comparisons run under the cross-encoder pipeline (Table~\ref{tab:rq2-ce}), alongside each pool's measured redundancy \eqref{eq:pool-redundancy}. 2WikiMultiHopQA reproduces the HotpotQA conclusions: every fixed-pressure diversifier is significantly below $k$-NN and the tuned \RNGSCORE sits marginally above it. Mild diversification pays on MuSiQue: $+2.1$ S-Recall points for parameter-free greedy \DPP and for \MMR($\lambda=0.9$) while $+2.5$ for the tuned \RNGSCORE. NQ-Open sits between, with small gains for the mildest operating points.

	The tuned $\gamma$ margins are read as in Section~\ref{sec:discussion} jointly with how often the induced selection deviates from $k$-NN. On 2WikiMultiHopQA, the selected $\gamma=-0.3$ is the fallback end of the sweep grid and reproduces $k$-NN on $98.8\%$ of queries: validation essentially says ``do not diversify''. On NQ-Open, a single-hop dataset, the mild negative margin acts on only $2.8\%$ of queries and, where it acts, improves S-Recall on twice as many queries as it degrades. On MuSiQue, the selected margin is mild ($\gamma=0.1$, with $\gamma=-0.1$ performing comparably). The penalty fires at ordinary inter-passage distances rather than at duplicate scale, deviating from $k$-NN on $16.5\%$ of queries. What that pressure buys tracks evidence multiplicity rather than duplicate structure: split by hop count, the tuned margin's gain grows from $+1.1$ S-Recall points on $2$-hop queries through $+3.3$ on $3$-hop to $+4.8$ on $4$-hop, where the $k$-NN baseline falls to $50.5$ as up to four gold paragraphs compete for five slots. \MMR($\lambda=0.9$) and parameter-free \DPP\ show the same gradient, and the $2$-hop rows behave like 2WikiMultiHopQA.

	\begin{table}[ht]
		\centering
		\small
		\caption{Single-stage reranking on the natural pools of the remaining QA datasets (S-Recall@5, \texttt{bge-m3}, single run each; red$_{0}$ = measured pool redundancy \eqref{eq:pool-redundancy}; tuned operating points $\lambda=0.9/0.9/0.7$ and $\gamma=-0.3/0.1/-0.1$ for 2WikiMultiHopQA/MuSiQue/NQ-Open). Best in bold, second best underlined.}
		\label{tab:single-stage}
		\setlength{\tabcolsep}{5pt}
		\begin{tabular}{lcccccc}
			\hline
			\textbf{Dataset} & \textbf{red$_{0}$} & \textbf{$k$-NN} & \textbf{Maxmin} & \textbf{\DPP} & \textbf{\MMR$^{*}$} & \textbf{\RNGSCORE$^{*}$} \\
			\hline
			2WikiMultiHopQA & 0.0006 & \underline{85.7} & 69.3 & 84.4 & 85.5 & \textbf{85.8} \\
			MuSiQue         & 0.0011 & 69.2 & 48.4 & \underline{71.3} & \underline{71.3} & \textbf{71.7} \\
			NQ-Open         & 0.0007 & 92.1 & 86.8 & 91.6 & \underline{92.8} & \textbf{92.9} \\
			\hline
		\end{tabular}
	\end{table}

\end{document}